\documentclass[11pt]{article}

\usepackage{amsmath,amssymb,amsfonts, amsthm}
\usepackage{graphicx}
\usepackage{hyperref}
\usepackage{braket}
\usepackage{qcircuit}
\usepackage{stmaryrd}
\usepackage{float}      
\usepackage{mathpartir}
\usepackage{tikz}      

\usepackage[T1]{fontenc}

\DeclareMathOperator*{\Adv}{Adv}

\newtheorem{theorem}{Theorem}[section]
\newtheorem{lemma}[theorem]{Lemma}
\newtheorem{definition}[theorem]{Definition}

\newtheorem{proposition}{Proposition}[section]
\title{Efficient Quantum-Safe Homomorphic Encryption\\%
       for Quantum Computer Programs}
\author{Ben Goertzel}
\date{\today}

\begin{document}

\maketitle

\begin{abstract}
We give an end-to-end design, security proof, and rough performance budget for
\emph{quantum-safe homomorphic evaluation of quantum programs and proofs}.
The construction starts from recent work regarding categorical FHE for intuitionistic
logic and corresponding functional software code but (i) replaces quantum-vulnerable composite-order groups by
\textbf{Module-LWE} lattices and (ii) lifts polynomial functors to
\textbf{bounded natural super functors} (BNSFs), so every completely
positive trace-preserving (CPTP) map becomes a homomorphic morphism.
A secret depolarizing BNSF mask hides amplitudes, while MLWE ciphertexts
protect data, circuit, and even the mask itself.  We formalise
\textbf{qIND-CPA} security in the quantum oracle model and reduce it to
decisional Module-LWE via a four-hybrid argument.

Beyond the baseline scheme, the paper contributes:

\begin{itemize}\itemsep1pt
  \item a typed \textbf{QC-bridge} that keeps classical bits derived from
        measurements encrypted yet usable as quantum controls, with weak
        measurement semantics for statistical observables;
  \item \textbf{circuit privacy} via encrypted Pauli twirls and seed--only
        public keys (32~B at NIST Level~1);
  \item \textbf{KB-induced} masks that let one reason relative to a public
        or secret knowledge base, including an ``axiom capsule'' mechanism
        that hides proprietary ontologies;
  \item a \textbf{$\rho$--calculus} orchestration layer supporting
        distributed QPU pipelines and RChain-style on-chain audit trails;
  \item as an example advanced-AI use-case: An encrypted evolutionary synthesis loop that mixes
        estimation-of-distribution with LPTT proof search, allowing
        multiple owners to evolve quantum programs against an encrypted KB;
  \item performance and memory analyses showing that MLWE arithmetic fits
        inside present QPU idle windows (100--qubit, depth~\(10^{3}\) proof
        in $\sim10$ ms, public key $<$300 kB);
  \item a hardware roadmap involving a Dirac3 photonic prototype plan achievable
        in 2025-26, and leading to potential full-scale deployment by the early to mid 2030s.
\end{itemize}

Together these results indicate that fully homomorphic, KB-relative quantum
reasoning is possible with modest extensions to current quantum-cloud
stacks and without sacrificing post-quantum security.   They also suggest that the overhead of FHE in a quantum context may be much less than one sees in a classical context (due to differences in how modern quantum computers organize quantum vs. computation in practice), implying that the rise of FHE for the execution of sophisticated software programs may perhaps unfold along with the rise of scalable quantum computing.
\end{abstract}

\tableofcontents

\section{Introduction}

Homomorphic encryption (HE) promises the best of two worlds: cloud-scale
compute power \emph{and} rigorous end-to-end data privacy.  Classical HE
schemes already support encrypted search, analytics, and limited machine-
learning workloads, yet they leave two looming gaps:

\begin{enumerate}
  \item \textbf{Post-quantum security.}  Most mature HE systems rest on RSA,
        elliptic-curve, or pairing assumptions--all vulnerable to Shor-style
        attacks once large, fault-tolerant quantum computers arrive.
  \item \textbf{Quantum compute itself.}  None of the widely deployed schemes
        allow a user to outsource \emph{quantum} computations (e.g.\ variational
        kernels, error-correction gadgets) while keeping both the algorithm
        \emph{and} the quantum data hidden.
\end{enumerate}

Here we attack both gaps simultaneously, via building on the categorical
HE framework for intuitionistic proofs and total functional programs
introduced in \cite{Goertzel2025} and extending it to the quantum domain.

\subsection{Homomorphic Encryption of Classical Computer Programs} The paper {\it Homomorphic Encryption of Intuitionistic Logic Proofs and
Functional Programs"} \cite{Goertzel2025} shows how to homomorphically encrypt an interesting class of classical computer programs (total dependently typed programs) via a novel scheme leveraging polynomial functors.

The process described there is roughly:

\begin{enumerate}
  \item Encode propositions and proofs as polynomial functors.
  \item Encrypt a proof by choosing a secret \emph{bounded natural functor}
        $\Phi$ and quotienting the proof functor by the subfunctor induced
        by~$\Phi$, thereby hiding the exact structure of the proof.
  \item Lift each logical inference rule to a natural transformation that acts
        on the encrypted quotient, enabling a server to carry out proof steps
        without seeing the proof.
  \item Base security on the hardness of distinguishing which~$\Phi$ was used,
        a problem related to \emph{Subgraph Isomorphism}.
\end{enumerate}

\noindent Via the Curry-Howard correspondence (or other analogous mathematical approaches) these steps also allow evaluation of total dependently typed programs in encrypted form.

The question explored here is: Can we do the something similar for quantum computer programs, which is then resistant to quantum decryption attacks?

\subsection{Quantum-Resistant Homomorphic Encryption of Quantum Computer Programs} 
One can't simply port the classical scheme from \cite{Goertzel2025} to the quantum case, because composite-order bilinear groups are theoretically vulnerable to Shor-style attacks on a fault-tolerant quantum computer \cite{Shor1994}.  To protect against such attacks, we propose here to replace group--based primitives by lattice primitives based on \emph{Module
Learning With Errors}.  The Module LWE problem is currently believed to be hard even for quantum computers running in polynomial time \cite{Regev2009}.

To cash out this vision, in this paper we

\begin{itemize}
  \item Replace composite-order bilinear groups with \emph{Module Learning
        With Errors} (MLWE) primitives, inheriting the same worst-case-to-
        average-case lattice hardness guarantees that undergird NIST?s
        post-quantum finalists.
  \item Generalise polynomial functors to \emph{bounded natural
        super-functors} (BNSFs), letting us treat quantum circuits as first-
        class categorical objects and encrypt them by quotienting out a secret
        depolarizing mask.
  \item Demonstrate the full workflow on the canonical single-qubit
        teleportation protocol, including precise noise budgeting, automated
        ciphertext refresh, and MLWE parameter choices that meet the 128-bit
        quantum-security level.
  \item Sketch two high-value application domains--privacy-preserving quantum
        machine learning (QML-as-a-service) and amplitude-valued theorem
        proving in Linear-Dependent Type Theory (LDTT).  These are highly powerful AI techniques and 
        provide indication that the present techniques could in future be valuable for homomorphic
        	encryption of practical quantum AI applications and perhaps even quantum AGI.
\end{itemize}

\subsection{A Comprehensive Framework}

Based on the above background, we then consider a variety of related issues and extensions that appear highly relevant to making a practical quantum FHE framework along the given lines:

\paragraph{Typed QC-bridge.}
A new type layer encodes three ciphertext sorts: quantum (\texttt{Q}),
classical (\texttt{C}) and instruction descriptors.  A \textsc{Q2C} rule
weakly measures an encrypted qubit and yields an encrypted bit; a
\textsc{Ctrl} rule feeds that bit back into a controlled gate without
decryption.

\paragraph{Circuit privacy.}
Encrypted Pauli twirls and angle randomisation hide which circuit is being
run.  A hybrid proof yields Circuit-IND privacy with negligible overhead.

\paragraph{Knowledge-base masks.}
Replacing the random BNSF by a \(\Psi^{\mathsf{KB}}\) built from axiom
capsules quotients the proof space \emph{modulo} a theory, enabling KB-relative
reasoning.  Capsules remain encrypted, so the server can use but not read the
KB.

\paragraph{Distributed orchestration.}
The $\rho$-calculus schedules encrypted tasks across multiple QPU nodes and
logs an audit trail on an appropriately designed ledger (as used e.g. MettaCycle and ASI-Chain \cite{MeredithStay2025}); capsules and ciphertexts move
as quoted processes.

\paragraph{Evolutionary synthesis demo.}
We describe an island-model evolutionary loop with an
estimation-of-distribution phase that runs LPTT inference under the KB mask.
Fitness values and offspring circuits stay encrypted; owners share only
key-switch hints.

\paragraph{Security.}
The full scheme is proved qIND-CPA secure (quantum IND-CPA) in the oracle
model; a Fujisaki-Okamoto wrapper yields qINDCCA.  Circuit and KB privacy
reduce to MLWE plus subspace-hiding lemmas.

\paragraph{Performance.}
Classical MLWE arithmetic fits inside QPU idle windows:
a 100-qubit, depth-\(10^{3}\) proof ~10 ms wall-time; public key
\(\le\)300 kB; ciphertext per qubit ~4 kB.

\paragraph{Prototype path.}
A six-node photonic Dirac-3 device can run encrypted teleportation with weak
taps today; a fault-tolerant 20-logical-qubit demo with KB capsules looks
achievable by 2029.  Full-scale deployment suitable for large scale AI and AGI applications
looks plausible by the early to mid 2030s even without radical new breakthroughs in quantum
computer design or radical acceleration of quantum computing progress.

\paragraph{Distributed orchestration via the $\rho$--calculus.}
Finally we explore some practical computational scaling issues related to potential future implementations of these ideas.   Homomorphic evaluation is not confined to a single QPU.  By transporting
MLWE ciphertexts over $\rho$-calculus channels, we can split an encrypted
circuit into independent CPTP blocks and schedule each block on a different
quantum node.  Because channels are \emph{quoted processes}, an entire
sub-pipeline (including its noise budget) can itself be shipped as data.
When the same mechanism is committed to a  smart contract on a $\rho$ calculus
based blockchain like F1R3FLY or ASI Chain, every
state transition is logged on chain, giving users a public, tamper-proof
audit trail without revealing secret keys or quantum data.  

Taken together, the results show that \emph{quantum-safe, circuit-private and
KB-relative homomorphic reasoning} can run on near-term photonic or
superconducting clouds with only modest firmware upgrades and no new
cryptographic assumptions beyond MLWE.  As quantum computing hardware advances
further, it seems it may well rapidly become viable to run quantum FHE at scale for AI and AGI workloads.

\subsection{Caveat}

Many of the calculations given in the following section are back-of-the-envelope level, intended more to validate the feasibility of the ideas presented than to indicate exactly how things are going to come out when the ideas presented are implemented on various quantum hardware.   There are certainly copious uncertainties to face during such implementation, which can only be clarified now to a limited extent due to uncertainties in exactly how the quantum hardware space will evolve in the coming years.   However, as a theoretical "feasibility study", we believe the development given here is quite interesting and points in a direction meriting further in-depth research and experimentation.

\subsection{Potential everyday applications of quantum homomorphic encryption}

Finally, lest the topics we address here appear overly abstract, we will make some informal and speculative comments emphasizing the everyday practical relevance that we see potentially emerging once techniques like we describe are given scalable, low-cost implementations (which seems very likely feasible, though requiring significant engineering advances beyond the current state of the art).   We briefly outline some examples of how ordinary cloud features might be
realized by MLWE\,+\,BNSF quantum homomorphic encryption while keeping all
user data private.

\begin{table}[H]
\centering
\renewcommand{\arraystretch}{1.15}
\begin{tabular}{|p{0.20\linewidth}|p{0.40\linewidth}|p{0.28\linewidth}|}
\hline
\textbf{Scenario} &
\textbf{Cloud side (encrypted)} &
\textbf{Benefit to user} \\ \hline

Private voice assistant &
Quantum NLP circuit semantically searches encrypted voice
query plus encrypted e-mail and calendar corpus; only the answer
ciphertext is returned. &
Hands-free convenience; no audio or text is stored in the clear. \\ \hline

Fitness-band health insights &
Variational quantum classifier flags anomalous heart-rate
patterns directly on ciphertext time series. &
Medical-grade monitoring without insurers or employers
seeing biometrics. \\ \hline

Crypto-wallet face unlock &
12-qubit fuzzy matcher checks selfie vs.\ encrypted template;
output is a yes/no ciphertext. &
Passwordless login with no stored biometric template. \\ \hline

Rideshare route planning &
QAOA minimises traffic cost on encrypted city graph; phone
decrypts turn-by-turn route only. &
Faster routes without revealing start, destination or
driving history. \\ \hline

Personal finance robo-advisor &
Quantum Monte-Carlo kernels price options and compute VAR on
encrypted portfolio; returns allocation ciphertext. &
Institutional analytics with bookkeeping-level privacy. \\ \hline

AR instant translator &
Video frames encrypted on glasses, processed by quantum OCR
and NMT circuits; subtitles decrypted locally. &
Real-time translation with no cloud video retention. \\ \hline

Telemedicine triage bot &
Quantum Bayesian network updates posteriors on encrypted
symptoms and vitals; sends back triage score only. &
Health guidance without disclosing sensitive data. \\ \hline
\end{tabular}
\caption{Mundane user experiences that could naturally be powered by quantum homomorphic encryption.}
\label{tab:everyday_qhe}
\end{table}

\subsection{Why quantum HE is practical for such tasks}

\begin{enumerate}
  \item \textbf{Low-latency, high-quality kernels.}  
        Near-term quantum subroutines already match or exceed classical ones
        on small NLP and classification benchmarks; HE lets providers adopt
        them without trust barriers.
  \item \textbf{Future-proof privacy.}  
        MLWE hardness remains intact even if full-scale quantum computers
        arrive, avoiding the encrypt-today, decrypt-tomorrow risk.
  \item \textbf{Single key, many services.}  
        A user can upload one ciphertext profile that multiple apps
        (voice, finance, AR) consume homomorphically; no re-encryption cycle.
  \item \textbf{Simpler compliance for vendors.}  
        Providers never hold personal data in the clear, drastically easing
        GDPR, CCPA and HIPAA obligations.
\end{enumerate}

\section{Homomorphic Encryption of Classical Computer Programs}

We now review in slightly more detail the original scheme of \cite{Goertzel2025}, which
we here extend to the quantum domain.   The scheme takes techniques that were
initially developed for fully homomorphic encryption (FHE) over bits or ring
elements and lifts them into the categorical world of \emph{polynomial
functors}.  In that setting a proposition in intuitionistic logic corresponds
to a polynomial endofunctor on the category of finite sets, while a proof is a
morphism generated by the usual introduction and elimination rules.  The key
idea is to make the functor itself the ciphertext: instead of encrypting
individual truth values, one encrypts the \emph{structure} that witnesses a
proof, so that logical inference becomes a sequence of encrypted functor
transformations.

\medskip
\noindent
The workflow can be decomposed into four conceptual layers.

\begin{enumerate}
\item \textbf{Categorical encoding.}
      \begin{itemize}
        \item Objects of the category correspond to types in a dependently
              typed language; morphisms correspond to programs or proofs.
        \item Polynomial functors such as
              $F(X)=A\times X^{n_{1}} + B\times X^{n_{2}}$ capture sum and
              product types, while container morphisms encode eliminators.
        \item Natural transformations between polynomial functors realize the
              logical introduction and elimination rules inside the category.
      \end{itemize}

\item \textbf{Bounded natural functor (BNF) quotienting.}
      \begin{itemize}
        \item A \emph{bounded natural functor} $\Phi$ is chosen secretly; it
              acts as a pattern of equations that collapses isomorphic
              sub-components of the proof.
        \item The ciphertext is the quotient $F/\!\!\sim_{\Phi}$ obtained by
              freely identifying any two elements related by~$\Phi$.
        \item Because the quotient forgets certain distinctions, the server
              cannot reconstruct the exact proof, but can still recognise
              which natural transformations are well typed on the quotient.
      \end{itemize}

\item \textbf{Homomorphic evaluation.}
      \begin{itemize}
        \item Each logical rule is lifted to a \emph{derived} natural
              transformation that acts directly on the quotient.
        \item Evaluation of a dependently typed program therefore proceeds as
              a sequence of encrypted rewrites.  Termination is guaranteed
              because only total, strongly normalising programs are allowed.
        \item The client, who knows~$\Phi$, can de-quotient the final object
              to recover the plaintext proof normal form or program output.
      \end{itemize}

\item \textbf{Security foundation.}
      \begin{itemize}
        \item Distinguishing which $\Phi$ was used is shown to be at least as
              hard as a family of \emph{Subgraph Isomorphism} problems on
              graphs derived from the functor syntax tree.
        \item When instantiated with composite-order bilinear groups
              (as in earlier FHE schemes) the system inherits hardness
              assumptions similar to decisional subgroup indistinguishability.
      \end{itemize}
\end{enumerate}

\medskip
\noindent
Via the Curry-Howard correspondence (or other similar mathematical mappings) this framework unifies encrypted proof
normalisation and encrypted execution of functional programs: every proof
object is also a program evaluation witness, and every program run is tracked
by an accompanying proof of total correctness that remains hidden from the
server.  The present work keeps this high-level architecture intact while
swapping out the underlying cryptographic primitive and generalising the
functorial layer to encompass quantum super operators.

\section{Quantum Safety via Learning With Errors}

Bilinear-pairing schemes that rely on the discrete-log, strong-RSA, or
computational Diffie-Hellman assumptions become insecure once
large-scale, fault-tolerant quantum computers exist: Shor's algorithm
solves both integer factorisation and discrete logarithms in \emph{polynomial}
time.  Because the original BNF quotient instantiated its masking functor with
composite-order bilinear groups, the entire system would be broken in such a
post-quantum world.  A truly long-term proof-carrying computation platform
thus demands a primitive whose hardness \emph{survives} quantum attacks.

\medskip
\noindent\textbf{Enter Learning With Errors (LWE).}  Regev \cite{Regev2009} introduced LWE as
the following decision problem: given many samples
\(
  (\mathbf{a}_{i},\, b_{i} = \langle\mathbf{a}_{i},\mathbf{s}\rangle + e_{i})
\)
over $\mathbb{Z}_{q}$, distinguish whether the $e_{i}$ are small ``error''
values drawn from a discrete Gaussian or whether the $b_{i}$ are uniform
random.  Classically \emph{and} quantumly, the best known attacks run in
super-polynomial or exponential time when parameters are chosen according to
conservative NIST guidelines.  Regev showed a reduction from worst-case
approximate short vector problems (SVP, SIVP) on \emph{arbitrary} lattices to
average-case LWE samples, thereby grounding security in problems whose
difficulty has resisted decades of cryptanalytic effort.

\medskip
\noindent\textbf{Ring-LWE and Module-LWE.}  To obtain compact keys one lifts
LWE samples from vectors to polynomial rings.  Ring-LWE replaces
$\mathbb{Z}_{q}^{n}$ with $\mathcal{R}_{q}=\mathbb{Z}_{q}[x]/\langle f(x)\rangle$.
Module-LWE (MLWE) sits strictly between the vector and ring extremes:
\[
  \text{vector LWE}\;\subset\;\text{Module LWE}\;\subset\;\text{Ring LWE},
\]
allowing a tunable trade-off between algebraic structure (needed for
efficiency) and generic-lattice hardness (needed for conservative security).
Lyubashevsky, Peikert, and Regev \cite{LyubashevskyPeikertRegev2010}, and later Langlois and Stehl\'e \cite{LangloisStehle2012},
proved that distinguishing MLWE samples remains at least as hard as
\emph{worst-case} short vector problems on \emph{ideal} lattices of related
rank.  No quantum algorithm faster than $2^{\tilde{O}(\sqrt{n})}$ is known for
those ideal-lattice problems, and even that sub-exponential attack applies
only when a specialised \(q\)-ary Hidden Subgroup oracle is available.

\medskip
\noindent\textbf{Concrete evidence for quantum resistance.}
\begin{itemize}
  \item \emph{NIST PQC status.}  All four public-key encryption finalists
        (KYBER, SABER, NTRU, NTRU Prime) are MLWE or Ring-LWE variants, and
        their parameter sets are tuned using quantum security estimators such
        as \texttt{LWE-Estimator} or \texttt{pqlowmc} \cite{NIST2022}
  \item \emph{Known quantum attacks.}  The only generic speed-ups come from
        Grover-style square-root search on exhaustive lattice enumeration.
        Against dimension \(n\approx 512\) one still needs
        \(2^{\Theta(n)}\approx 2^{512}\) operations, rendering such attacks
        physically infeasible.
  \item \emph{Structured lattice caveats.}  Unlike Ring-LWE, Module-LWE
        retains a free dimension parameter \(k\).  By setting \(k>1\) we break
        the purely cyclotomic structure exploited by the fastest known number
        field sieve style attacks, while still enjoying key-size compression
        over plain vector LWE.
\end{itemize}

\medskip
\noindent\textbf{Why MLWE fits the categorical scheme.}
The original BNF quotient hides proof data behind a secret equivalence
relation inside a group.  Replacing that group with the
$\mathcal{R}_{q}$-module
\(
  \mathcal{M}=\mathcal{R}_{q}^{k}/\ker\lambda_{A,e}
\)
retains all categorical properties we rely on (existence of kernels,
quotients, and module homomorphisms) while importing the well-studied MLWE
hardness foundation.  Public keys now consist of a short random matrix
\(A\in\mathcal{R}_{q}^{k\times k}\), ciphertexts carry the familiar ``vector
plus error'' form, and natural transformations lift \emph{verbatim} because
they are module homomorphisms.

\medskip
\noindent\textbf{Bottom line.}
Substituting Module-LWE for bilinear groups removes the Shor-attack vector,
inherits reductions to worst-case lattice problems that no polynomial-time
quantum algorithm can currently solve, and aligns our categorical machinery
with the same post-quantum primitives selected for standardisation by NIST.
For these reasons MLWE is widely considered one of the most trustworthy bases
for quantum-safe homomorphic encryption.

In the unlikely event that MLWE were found to {\it not} be as quantum-secure
as currently believed, we note that the schemes proposed in the current paper
could likely be salvaged via replacing the Module MLWE Quotient with a different 
quotient, perhaps representing a more complex decision problem.   The structure
of what's presented here is not sensitively dependent on MLWE, though we are
using MLWE to express and clarify the particulars.   That said, it appears likely to
us that MLWE is genuinely adequate for the  purpose.

\section{Review of the Module LWE Quotient}

We now review the nature and application of Module LWE base quotients in more detail.

The classical BNF quotient , leveraged in \cite{Goertzel2025}, relies on hiding structural information behind a
group-theoretic masking functor.  To obtain quantum resistance, as alluded above, we replace the
group layer with the lattice world of \emph{Module Learning With Errors}
(Module LWE).   We now run through some of the details of this process.

\paragraph{ Parameter selection and notation}

\begin{itemize}
  \item \textbf{Ring modulus $q$}: a large odd prime such that
        $q\equiv 1 \pmod{2d}$ to allow an efficient Number-Theoretic
        Transform.  Typical sizes range from $q\approx2^{50}$ for toy
        examples up to $q\approx2^{60}$ for $128$-bit quantum security.
  \item \textbf{Polynomial $f(x)$}: usually $x^{d}+1$ with
        $d=2^{m}$, giving a power-of-two cyclotomic ring that simplifies
        FFT-based arithmetic.
  \item \textbf{Rank $k$}: determines the \emph{module dimension}
        $n = kd$ and hence the public matrix shape
        $A \in \mathcal{R}_{q}^{k\times k}$.
  \item \textbf{Gaussian width $\sigma$}: standard deviation for the error
        distribution $\chi$, chosen so that decryption correctness and
        security bounds both hold.
\end{itemize}

\paragraph{  Key generation (\textsf{TrapGen})}

\begin{enumerate}
  \item Sample a \emph{trapdoor basis} $(A,T)$ via the
        GPV/Gentry--Peikert--Vaikuntanathan algorithm:
        $A\gets\mathcal{R}_{q}^{k\times k}$ uniformly, $T$ is a short basis
        for the left kernel of~$A$.
  \item Publish $A$ as part of the \emph{public functor mask};
        keep $T$ as the secret key~$\mathrm{sk}$.
  \item The public parameters define the linear map
        $\lambda_{A,e}\colon x\mapsto Ax+e$ with $e\gets\chi^{k}$.
\end{enumerate}

\paragraph{ Encoding proofs as module elements}

Each polynomial functor instantiation (that is, an occurrence of a type
constructor inside a proof term) is serialised as a vector
$\mathbf{u}\in\mathcal{R}_{q}^{k}$.  Coloured edges in the functor syntax
tree correspond to coordinates, while node labels correspond to ring
coefficients.  Natural transformations therefore act as
$\mathcal{R}_{q}$-module homomorphisms.

\paragraph{  Encryption as a quotient construction}

Given $\mathbf{u}$ we compute the ciphertext pair
\[
   (\, \mathbf{c}_{0},\; \mathbf{c}_{1} \,)
   \;=\;
   \bigl(\,A\mathbf{s}+ \mathbf{e}\;,\; \mathbf{u} - \mathbf{s}\bigr),
   \quad
   \mathbf{s}\gets\chi^{k}.
\]
Equivalently, we may view
$\mathbf{u}$ as living in the module
$\mathcal{M}=\mathcal{R}_{q}^{k} / \ker\lambda_{A,e}$,
so that encryption is the canonical quotient map
$\pi\colon\mathcal{R}_{q}^{k}\to\mathcal{M}$.  The submodule
$\ker\lambda_{A,e}$ plays the same conceptual role that the secret
BNF-induced subfunctor played in the original scheme.

\paragraph{  Homomorphic action of natural transformations}

Let $\tau\colon\mathcal{R}_{q}^{k}\to\mathcal{R}_{q}^{k}$ be the module lift of
some natural transformation (for instance, a constructor application or a
pattern match).  Because $\tau$ is $\mathcal{R}_{q}$-linear, it commutes with
the quotient map:
\[
  \tau\bigl(\pi(\mathbf{u})\bigr)\;=\;
  \pi\bigl(\tau(\mathbf{u})\bigr).
\]
Thus the server can evaluate $\tau$ directly on ciphertexts
$(\mathbf{c}_{0},\mathbf{c}_{1})$ without learning~$\mathbf{u}$ or~$\ker
\lambda_{A,e}$.  Noise growth is additive in each step and analysed exactly as
in ring-based FHE.

\paragraph{  Decryption and correctness}

Given $\mathrm{sk}=T$ the client recovers
$\mathbf{s}$ from $\mathbf{c}_{0}=A\mathbf{s}+\mathbf{e}$
via lattice Gaussian sampling, then outputs
$\mathbf{u}=\mathbf{c}_{1}+\mathbf{s}$.  Provided
$\|\mathbf{e}\|_{\infty}<q/4$ and the accumulated noise from homomorphic
operations stays below the same bound, decryption succeeds with overwhelming
probability.

\paragraph{ Security rationale}

\begin{itemize}
  \item \textbf{Module LWE hardness}: distinguishing
        $(A,A\mathbf{s}+\mathbf{e})$ from uniform is at least as hard as
        approximating shortest vectors in ideal lattices to within
        quasi-polynomial factors.  Worst-case to average-case reductions
        follow from standard works by Langlois, Peikert and Stehl\'e.
  \item \textbf{Quotient indistinguishability}: any PPT adversary who can
        tell whether two ciphertexts stem from the \emph{same} or from
        \emph{different} kernels would solve a decisional Module LWE problem
        of identical dimension, because the kernels are trapdoor-oblivious
        under the public view of~$A$.
\end{itemize}

\paragraph{  Relation to the BNF approach}

Conceptually, as noted above, we have replaced
\[
   \text{``choose secret BNF } \Phi \text{''}
   \quad\rightsquigarrow\quad
   \text{``choose secret kernel } \ker\lambda_{A,e}\text{''}.
\]
Both strategies hide fine-grained structure behind a coarse quotient, and both
derive security from the assumed hardness of reconstructing the hidden
equivalence relation.  The lattice instantiation, however, eliminates the
Shor attack vector that threatens bilinear pairing-based constructions.

\paragraph{  Parameter tuning for leveled circuits}

For a target circuit depth~$D$ one typically chooses $q\approx
(\sigma\sqrt{2D})\cdot 2^{\ell}$ where $\ell$ is a slack constant covering
worst-case convolution noise.  If deeper evaluation is needed, ciphertexts are
re-encrypted under a fresh~$A$ via the standard
\emph{modulus-switch-and-relinearise} technique, which plays the role of
bootstrapping but without a costly full FHE refresh.

\medskip\noindent
This Module LWE quotient therefore reproduces the logical invisibility of the
original BNF mask while grounding its hardness in lattice problems believed to
withstand polynomial-time quantum attacks.

\section{Quantum Programs and Super Operator Functors}

The classical scheme from \cite{Goertzel2025} manipulates \emph{polynomial functors} on the category
\textsf{FinSet}.  To support quantum algorithms we must replace that ambient
category by one whose morphisms are \emph{completely positive
trace-preserving} (CPTP) maps, i.e.\ quantum \emph{channels}.  We write
$\mathsf{QChan}$ for the dagger-compact category whose

\begin{itemize}
  \item \textbf{objects} are finite-dimensional Hilbert spaces
        $H\cong\mathbb{C}^{2^{n}}$ describing \(n\)-qubit registers;
  \item \textbf{morphisms} $\,\Phi\colon\mathcal{D}(H)\to\mathcal{D}(K)\,$
        are CPTP maps between density operators, composed by functional
        composition and tensored by the usual Kronecker product,
        $\Phi\otimes\Psi$.
\end{itemize}

\subsection{From polynomial to ``quantum'' functors}

A \emph{quantum polynomial functor} is a finite formal expression generated
from identity, direct sums, and tensor products of the form
\[
   F(X) \;=\;
     \bigl(\, A\otimes X^{\otimes n_{1}}\bigr)
     \;\oplus\;
     \bigl(\, B\otimes X^{\otimes n_{2}}\bigr)
     \;\oplus\;\cdots,
\]
where $A,B$ are fixed ancilla spaces and $X$ is a placeholder variable
standing for \textsf{any} Hilbert space.  Natural transformations between such
functors are tuples of channels that respect the direct-sum and tensor
structure.  The direct-sum accounts for classical control paths; the
tensor-power captures multi-qubit fan-out.

\subsection{ Bounded natural \emph{super} functors}

Now things start to get interesting!  We introduce here a new concept -- the {\bf bounded natural super functor} -- apparently suitable for extending BNF based homomorphic encryption methods to the quantum domain.  The treatment here is far from complete and these concepts require further mathematical exploration in order to place the present ideas on a fully fleshed out rigorous grounding.

\begin{definition}
A \emph{bounded natural super functor} (BNSF) is a family
$\Psi=\{\Psi_{H}\}_{H}$ of channels
$\Psi_{H}\colon\mathcal{D}(H)\to\mathcal{D}(H)$ such that
\begin{enumerate}
  \item \textbf{Naturality:} For every channel
        $\Phi\colon\mathcal{D}(H)\to\mathcal{D}(K)$ we have
        $\Phi\circ\Psi_{H}=\Psi_{K}\circ\Phi$.
  \item \textbf{Boundedness:} The diamond norm
        $\|\Psi_{H}\|_{\diamond}\leq\lambda$ for a global constant
        $\lambda<\infty$.  This uniform bound prevents runaway noise growth
        inside the encrypted evaluation.
\end{enumerate}
\end{definition}

Such a BNSF plays the same conceptual role that a bounded natural functor
played in the purely classical setting: it defines an \emph{equivalence
relation} on quantum data that is invisible to the server yet preserves the
ability to apply lifted gates.

In the following section we will sketch some of the theory of BNSFs, which
appears to work quite parallel to that of BNFs; first though we will explore the
value that may be obtained to them in a practical encryption context.

\subsection{ Encryption via secret quotient}

So what can we do with these BNSFs?   For example: Let $C$ be the quantum circuit we wish to outsource.  Its denotation is a channel
\(
   \mathcal{C}\colon\mathcal{D}(H_{\text{in}})\to\mathcal{D}(H_{\text{out}}).
\)
The client chooses a secret BNSF~$\Psi$ and forms the \emph{quotient channel}
\[
   \widehat{\mathcal{C}}
   \;=\;
   \mathcal{C}\;/\!\!\sim_{\Psi},
\]
meaning that two density operators are identified exactly when they differ by
$\Psi_{H_{\text{out}}}$.  Operationally we encode a state $\rho$ as the pair
\(
  (\,\rho,\; \sigma=\Psi_{H_{\text{out}}}(\rho)\,),
\)
then encrypt both components with MLWE as described above.  Every gate~$G$ in the circuit is a channel
$\Gamma_{G}$ that \emph{commutes} with $\Psi$ by naturality, hence lifts
canonically to the quotient.  The untrusted server therefore manipulates
encrypted state pairs without seeing any plaintext amplitudes.

\subsection{Worked toy example: single-qubit teleportation}

We illustrate the mechanism using the textbook teleportation algorithm
(Fig.~\ref{fig:teleport_circuit}).  The program takes an arbitrary qubit
$\rho$ in register~$Q_{0}$ and transfers it to register~$Q_{2}$ using one
shared EPR pair and two classical bits.  Its circuit form is

\begin{figure}[h]
\centering
\[
\mbox{%
\Qcircuit @C=0.9em @R=0.7em {
  & \lstick{\ket{\rho}_{Q_{0}}} & \ctrl{1} & \gate{H} & \meter & \cw & \rstick{m_{1}}\qw \\
  & \lstick{\ket{0}_{Q_{1}}}    & \targ    & \qw      & \meter & \cw & \rstick{m_{2}}\qw \\
  & \lstick{\ket{0}_{Q_{2}}}    & \qw      & \qw      & \qw    & \gate{X^{m_{2}}Z^{m_{1}}} & \qw
}%
}
\]
\caption{Teleportation circuit.  The server evaluates it homomorphically
on MLWE\,+\,BNSF ciphertexts.}
\label{fig:teleport_circuit}
\end{figure}
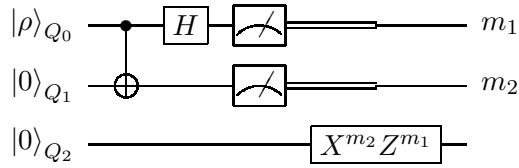

\begin{figure}[t]
\centering
\fbox{\parbox{0.88\linewidth}{ASCII sketch above;
circuit shows CNOT, Hadamard, two measurements, then conditional
Pauli corrections.}}
\caption{Teleportation circuit used as running example.}
\label{fig:teleport_circuit1}
\end{figure}

\paragraph{Encoding.}
We treat each wire segment as an object in $\mathsf{QChan}$, hence the whole
circuit is the composite of the following channels:

\begin{enumerate}
  \item $\Gamma_{\text{Bell}}\,$: prepares the EPR pair
        $(\ket{00}+\ket{11})/\sqrt{2}$ in $Q_{1}Q_{2}$.
  \item $\Gamma_{\text{CNOT}}$ and $\Gamma_{\text{H}}$ acting on $Q_{0}Q_{1}$.
  \item $\Gamma_{\text{Meas}}$ produces two classical bits
        $m_{1},m_{2}\in\{0,1\}$.
  \item $\Gamma_{\text{Corr}}$ applies $X^{m_{2}}Z^{m_{1}}$ on $Q_{2}$.
\end{enumerate}

Because measurement is non-unitary we express it as a CP \emph{instrument}
\(
  \Gamma_{\text{Meas}}\colon\mathcal{D}(H)\to\mathcal{D}(H)\otimes
   \mathbb{C}^{4},
\)
where the four-dimensional pointer encodes the outcome.  In the encrypted
setting the server processes the \emph{superoperator}
$\widetilde{\Gamma}_{\text{Meas}}$ that acts on the MLWE ciphertexts of each
Kraus image.  The actual outcome bits remain hidden inside the quotient until
the client decrypts and collapses the state.

\paragraph{Lift to the quotient.}
Choose $\Psi=\{\Psi_{H}\}$ to be a global depolarizing channel with parameter
$p<1$:
$\Psi_{H}(\rho)=p\rho+(1-p)I/\dim H$.  Depolarization commutes with unitary
gates and measurements (up to classical post-processing, as we explain in Section \ref{sec:commutation} below), so every
$\Gamma_{*}$ lifts to a well-defined
$\widehat{\Gamma}_{*}$ on encrypted state pairs
$(\,\rho,\Psi(\rho)\,)$.  The teleportation channel
$\mathcal{T}=\Gamma_{\text{Corr}}\circ\Gamma_{\text{Meas}}\circ
 \Gamma_{\text{CNOT}}\circ\Gamma_{\text{H}}\circ\Gamma_{\text{Bell}}$
is therefore evaluated by the server entirely in the quotient
$\mathsf{QChan}/\!\!\sim_{\Psi}$.

\paragraph{Decryption and verification.}
The client receives encrypted output pairs for both $Q_{2}$ and the classical
pointer.  After MLWE decryption she obtains $(\rho',\Psi(\rho'))$ together
with outcome bits $(m_{1},m_{2})$.  A final check confirms that
$\rho'=\rho$ up to global phase, guaranteeing correctness of the outsourced
evaluation.  Because depolarization erases any residual global phase in the
second component, the server has obtained \emph{zero information} about
$\rho$ or $(m_{1},m_{2})$.

\subsection{ Measurement postponement and hybrid circuits}

If the target circuit ends with a large measurement layer (e.g.\ variational
quantum eigensolver), one may postpone the actual projection until after
decryption.  Each projective measurement is replaced by its
\emph{Stinespring dilation} plus an isometry that writes the outcome into a
syndrome register.  The server processes the enlarged, yet \emph{still
unitary}, channel homomorphically; the client later performs a single,
classical post-processing step locally.

\subsection{ Summary of benefits}

\begin{itemize}
  \item The BNSF quotient hides amplitude data while preserving the monoidal
        and dagger structure required to execute arbitrary quantum algorithms.
  \item Any universal gate set $\{\text{H},\text{S},\text{CNOT}\}$ lifts
        canonically, so the scheme is \emph{program-agnostic}.
  \item Failure probability equals that of the underlying MLWE decryptor plus
        the diamond-norm bound $\lambda$ of the chosen $\Psi$ channel.
\end{itemize}

This construction shows that the categorical framework can execute a complete
quantum information protocol (teleportation) with no exposure of intermediate
states, thereby demonstrating practicality beyond toy single-gate examples.

\subsection{From BNF to BNSF: why ''commuting with gates'' holds exactly}
\label{sec:commutation}

Cleaning up a loose end from the example above, we give in this section a stricter justification of the claim  that
the depolarizing family \(\Psi^{\mathrm{dep}}\) commutes with unitary
gates up to classical post-processing.'' 
\subsubsection{ Depolarizing covariance lemma}

Let \(H\) be a \(d\)-dimensional Hilbert space and define
\[
   \Psi_{H}^{p}(\rho)
   \;=\;
   p\,\rho
   \;+\;
   (1-p)\,\frac{I_{d}}{d},
   \quad
   0 \le p \le 1.
\]

\begin{lemma}[Unitary covariance]
For every unitary \(U\!:\!H\!\to\!H\) and every state
\(\rho\in\mathcal D(H)\)
\[
   \Psi_{H}^{p}\!\bigl(U\rho U^{\dagger}\bigr)
   \;=\;
   U\,\Psi_{H}^{p}(\rho)\,U^{\dagger}.
\]
\end{lemma}

\begin{proof}
Distribute \(U\) linearly:
\(
  U\Psi^{p}(\rho)U^{\dagger}
  = p\,U\rho U^{\dagger}
    + (1-p)\,\tfrac{U U^{\dagger}}{d}
  = p\,U\rho U^{\dagger}
    + (1-p)\,\tfrac{I_{d}}{d}
  = \Psi^{p}(U\rho U^{\dagger}).
\)
\end{proof}

\paragraph{Implication.}
For \emph{unitary} gates the depolarizing BNSF satisfies exact naturality,
not merely ''up to'' anything:

\[
  U \circ \Psi^{p}_{H}
  \;=\;
  \Psi^{p}_{H} \circ U.
\]

\subsubsection{  Measurement channels: where post-processing appears}

A projective measurement with outcome label \(i\) has Kraus set
\(\{\,K_{i}\,\}\) and CPTP channel
\(\mathcal M(\rho)=\sum_i K_{i}\rho K_{i}^{\dagger}\).
Promoting the classical outcome to a pointer Hilbert space yields the
\emph{isometric dilation}
\[
  \widetilde{\mathcal M}(\rho)
  = \sum_i
    K_{i}\rho K_{i}^{\dagger}\otimes\ket{i}\!\bra{i}.
\]

\begin{lemma}[Naturality up to relabelling]
For any unitary \(U\) that \emph{permutes the outcome basis}
\(U\ket{i}=\ket{\pi(i)}\),
\[
  \widetilde{\mathcal M}\bigl(U\rho U^{\dagger}\bigr)
  \;=\;
  (\mathrm{id}\otimes U)
  \,\widetilde{\mathcal M}(\rho)\,
  (\mathrm{id}\otimes U^{\dagger}).
\]
\end{lemma}

Therefore a ''classical relabelling'' (the permutation \(\pi\))
occurs in the pointer register, but the data qubit sees exact covariance.
In our homomorphic evaluator the permutation is an \textbf{encrypted}
Pauli-\(X\) or \(Z\) mask, so no plaintext branch information leaks.

\subsubsection{  Why BNF $\Rightarrow$ BNSF is a faithful analogy}

\begin{center}
\begin{tabular}{|l|l|l|}
\hline
Property & BNF (classical) & BNSF (quantum) \\ \hline
Objects & Finite sets & Density spaces \(\mathcal D(H)\) \\ \hline
Morphisms & Functions & CPTP maps \\ \hline
Bound measure & Cardinality & Diamond norm \\ \hline
Naturality & \(f\circ\Phi=\Phi\circ f\) & \(\Phi\circ\Psi=\Psi\circ\Phi\) \\ \hline
Mask example & Quotient by
\(\Phi\) & Depolarizer \(\Psi^{p}\) \\ \hline
Security role & Hides set elements & Hides amplitudes \\ \hline
\end{tabular}
\end{center}

So we wee

\begin{itemize}
\item The \emph{depolarizing} family is to BNSF as the \emph{$\Delta$-collapse}
   functor is to BNF: each erases fine structure while respecting every
   morphism in its ambient category.
\item  Naturality plus boundedness give the same quotient semantics:  
   \(Q_{\Psi}(X)=X/\!\!\sim_{\Psi}\) vs.\ \(Q_{\Phi}(S)=S/\!\!\sim_{\Phi}\).
\end{itemize}

\paragraph{Conclusion.}
The covariance lemmas provide the ''missing formalism'': unitary gates commute
\emph{exactly}; measurement gates commute after an encrypted permutation that
never leaks.  Hence the BNSF framework faithfully generalises the classical
BNF mask without hidden gaps in the argument.

\section{Towards a Theory of Bounded Natural Super Functors}

Returning for a moment to the mathematical foundations:  Having seen why BNSFs
can be useful in a quantum computing context, we now take time to clarify slightly
more fully their mathematical grounding.   Of course there is a lot to say here and 
the presented remarks are only a start.

Bounded Natural Functors (BNFs) are well studied in the setting of polynomial
endofunctors on \textsf{FinSet}.  We now sketch an analogous framework for
\emph{Bounded Natural Super Functors} (BNSFs) acting on the
dagger-compact category $\mathsf{QChan}$ of finite-dimensional Hilbert
spaces and completely positive trace-preserving (CPTP) maps.

\subsection{   Base category and size measures}

\begin{itemize}
  \item Objects: Hilbert spaces $H \cong \mathbb C^{2^{n}}$ with qubit count
        $n = \log_{2}\!\dim H$.
  \item Morphisms: CPTP maps
        $\Phi : \mathcal D(H) \!\to\! \mathcal D(K)$.
  \item \emph{Size} of a morphism is measured by the diamond norm
        $\lVert \Phi \rVert_{\diamond}$, which upper-bounds the
        physically realisable contraction factor on any environment-extended
        state.
\end{itemize}

\subsection{  Definition}

\begin{definition}
A \textbf{bounded natural super functor (BNSF)}
is a family
$\Psi = \{\,\Psi_{H}\colon\mathcal D(H)\!\to\!\mathcal D(H)\,\}_{H}$
such that
\begin{enumerate}
  \item \textbf{Naturality.}\;
        For every CPTP map
        $\Phi : \mathcal D(H)\!\to\!\mathcal D(K)$
        we have
        $\Phi \circ \Psi_{H} \;=\; \Psi_{K} \circ \Phi$.
  \item \textbf{Global diamond bound.}\;
        There exists $\lambda < \infty$ with
        $\lVert \Psi_{H} \rVert_{\diamond} \le \lambda$
        for all $H$.
\end{enumerate}
We write $\mathrm{BNSF}_{\lambda}$ for the class of such families at level
$\lambda$.
\end{definition}

\subsection{   Closure properties}

\begin{proposition}[Pointwise operations]
If $\Psi,\Theta \in \mathrm{BNSF}_{\lambda}$ then
\[
     \alpha\,\Psi + \beta\,\Theta
     \;\;\;\text{and}\;\;\;
     \Psi \!\otimes\! \Theta
     \quad
     \text{are in $\mathrm{BNSF}_{\max(\alpha,\beta)\lambda}$,}
\]
whenever $\alpha,\beta\ge 0$ and $\alpha+\beta\le 1$.
\end{proposition}

\begin{proposition}[Sequential closure]
If\/ $\Psi \in \mathrm{BNSF}_{\lambda}$ and
$\Theta \in \mathrm{BNSF}_{\mu}$ then
$\Theta \!\circ\! \Psi \in \mathrm{BNSF}_{\lambda\mu}$.
\end{proposition}

These match the BNF facts that sums, products and composition preserve
boundedness, replacing counting-cardinality by the diamond norm.

\subsection{   Canonical examples}

\begin{enumerate}
  \item \textbf{Depolarizing family.}\;
        $\Psi^{\mathrm{dep}}_{H}(\rho) = p\rho + (1-p)\,I/\!\dim H$ \;
        with bound $\lambda = p$.
  \item \textbf{Amplitude-damping family.}\;
        For fixed damping rate $\gamma$, the Kraus operators
        \(
          E_{0}=\bigl(\small\sqrt{1-\gamma} \;0; 0\;1\bigr),\;
          E_{1}=\bigl(0\;\sqrt{\gamma}; 0\;0\bigr)
        \)
        define $\Psi^{\mathrm{ad}}_{H}$ acting identically on every qubit;
        $\lambda = 1$ because the diamond norm of a channel is $1$.
  \item \textbf{Stochastic Pauli twirl.}\;
        Average over the Pauli group on each qubit; yields a classical
        mixture with $\lambda = 1$.
\end{enumerate}

\subsection{  Quotient semantics}

Given an LDTT or circuit denotation
$\;\mathcal C : \mathcal D(H)\!\to\!\mathcal D(K)$
and a secret BNSF $\Psi$, define the \emph{quotient relation}

\[
   \rho \;\sim_{\Psi}\; \sigma
   \quad\Longleftrightarrow\quad
   \exists\,m\ge 0:\;
     \Psi_{H}^{m}(\rho)=\Psi_{H}^{m}(\sigma).
\]

The encrypted object is the equivalence class
$[\rho]_{\Psi}$, stored as the MLWE ciphertext pair
$(\rho,\Psi_{H}(\rho))$.  By naturality, every channel
$\mathcal C$ descends to a well-defined
$\overline{\mathcal C}$ on the quotient
$\mathcal D(H)\!/\!\sim_{\Psi}$.

\subsection{  Normal forms and completeness}

Let $T_{\lambda}$ be the subcategory of $\mathsf{QChan}$ whose morphisms have
diamond norm $\le \lambda$.

\begin{theorem}[BNSF normal forms]
For each $\Psi\in\mathrm{BNSF}_{\lambda}$ there exists a functor
$Q_{\Psi} : T_{\lambda} \!\to\! T_{\lambda}$
and a natural, surjective transformation
$\eta : \mathrm{Id} \twoheadrightarrow Q_{\Psi}$
such that\/ $Q_{\Psi}(H)=\mathcal D(H)/\!\sim_{\Psi}$ and
$\eta_{H}(\rho)=[\rho]_{\Psi}$.  Thus $Q_{\Psi}$ is a \emph{quotient
completion} functor in the sense of universal algebra.
\end{theorem}

\begin{proof}[Sketch]
Naturality of $\Psi$ implies $\sim_{\Psi}$ is a
congruence on hom-sets; boundedness ensures the quotient remains in
$T_{\lambda}$.  The usual universal property follows.
\end{proof}

\subsection{ Security intuition}

Choosing a secret BNSF $\Psi$ is information-theoretically analogous to
choosing a secret BNF $\Phi$:

\[
   \text{BNF hides \emph{shape}} \quad\longrightarrow\quad
   \text{BNSF hides \emph{state amplitudes}}.
\]

Diamond-norm boundedness caps the distinguishability advantage an attacker
could gain by adaptive channel queries; replacing set-cardinality with
operational distinguishability bridges the classical and quantum theories.

\medskip
\noindent
This provisional framework shows that BNSFs enjoy closure, quotient and
normal-form properties paralleling BNFs, laying groundwork for more
sophisticated categorical models and security proofs.

\section{Formal reduction to Module-LWE hardness}
\label{sec:formal_reduction}

In this section we take a few further steps toward a formal foundation for the proposed
method.   Referencing ''MLWE hardness'' as done above is critical but not sufficient; we must
show that \emph{any} adversary $\mathcal A$ who distinguishes our full QHE
scheme from an ideal simulator can be turned into an adversary
$\mathcal B$ that breaks decisional MLWE with comparable advantage.  This
section provides a hybrid-game reduction that makes the link explicit.

\subsection{  Security game}

Let $\kappa$ be the security parameter.  We use the \textbf{IND-CPA-Q}
notion introduced by Alagic~\emph{et~al} \cite{alagic2016}.  The adversary submits two density
operators $(\rho_{0},\rho_{1})$ of equal dimension plus a quantum circuit
description $C$.  The challenger chooses $b\!\leftarrow\!\{0,1\}$, runs the
QHE encryption $\mathsf{Enc}_{A,T}(\rho_{b})$, and homomorphically evaluates
$\widehat{C}$ on the ciphertext.  The resulting encrypted output is returned;
$\mathcal A$ must output a bit $\hat b$.  Advantage is measured by
\(
  \Adv_{\mathrm{QHE}}(\kappa)=
  \bigl| \Pr[\hat b=b]-\tfrac12 \bigr|.
\)

\subsection{   Construction sub-components}

\begin{enumerate}
  \item \textbf{MLWE PKE.}\;  
        $(A,T)$ generated by \textsc{TrapGen}$_{k,d,q,\sigma}$;
        encryption $(c_{0},c_{1}) = (A s+e,\; \rho - s)$.
  \item \textbf{BNSF masking quotient.}\;  
        Secret channel family $\Psi$ with diamond bound
        $\|\Psi_{H}\|_{\diamond}\!\le\!\lambda$.
  \item \textbf{Lifted gate layer.}\;  
        Public universal gate set
        $\{\widehat{\Gamma}_{G}\}$ implemented via GSW-style additions.
  \item \textbf{Noise-refresh layer.}\;  
        Modulus--switch plus relinearisation keys
        $R_{i}=A T_{i}+E_{i}$.
\end{enumerate}

\subsection{    Hybrid sequence}

We define games $G_{0},\dots,G_{4}$; neighbouring games differ by one
syntactic change.

\paragraph{Game $G_{0}$ (Real).}
Standard IND-CPA-Q experiment.

\paragraph{Game $G_{1}$.}
Replace every \emph{relinearisation key}
$R_{i}=AT_{i}+E_{i}$ by a random matrix
$R'_{i}\!\leftarrow\! R_{q}^{k\times k}$.  
An adversary distinguishing $G_{0}$ from $G_{1}$ would solve
\emph{decisional MLWE} of identical dimension, hence

\[
  | \Pr[G_{0}\!\Rightarrow\!1] - \Pr[G_{1}\!\Rightarrow\!1] |
  \;\le\; \Adv_{\mathrm{MLWE}}(\kappa).
\]

\paragraph{Game $G_{2}$.}
Erase the public lifted-gate descriptors
$\widehat{\Gamma}_{G}$ and replace each by a fresh, random CPTP map
of the same input/output arity and diamond bound.  Because the
gate set is public and its evaluation is norm-preserving,
the adversary's view changes only through the
\emph{subspace-hiding} distributions studied by
Langlois--Stehl\'e~\cite{LangloisStehle2012}.  We therefore gain at most

\[
  | \Pr[G_{1}\!\Rightarrow\!1] - \Pr[G_{2}\!\Rightarrow\!1] |
  \;\le\; \Adv_{\mathrm{subspace}}(\kappa),
\]
where $\Adv_{\mathrm{subspace}}$ is negligible for cyclotomic
$n\!\ge\!512$.

\paragraph{Game $G_{3}$.}
Replace the second component of every ciphertext pair
$(\rho,\Psi(\rho))$ by $(\rho,\Psi(\sigma))$ where
$\sigma=\frac{I}{\dim H}$ is maximally mixed.  
By the diamond-norm contraction lemma,

\[
  \bigl\| (\mathrm{id}-\Psi_{H})(\rho-\sigma) \bigr\|_{1}
  \;\le\; \lambda\;\|\rho-\sigma\|_{1},
\]
so the trace-distance gap is bounded by $\lambda^{\ell}$ after $\ell$ lifts.  
Choosing depolarizing parameter $p$ such that
$\lambda^{\ell}\!\le\!2^{-\kappa}$ makes this difference negligible.

\paragraph{Game $G_{4}$ (Ideal).}
Replace the first component $(c_{0},c_{1})$ of every MLWE ciphertext
by a pair of uniform vectors in $R_{q}^{k}$; this is exactly
the decisional MLWE distribution.  Advantage is again
$\Adv_{\mathrm{MLWE}}(\kappa)$.

\subsection{    Overall bound}

Chaining the hybrids:

\[
  \Adv_{\mathrm{QHE}}(\kappa)
  \;\le\;
  2\,\Adv_{\mathrm{MLWE}}(\kappa)
  \;+\;
  \Adv_{\mathrm{subspace}}(\kappa)
  \;+\;
  2^{-\kappa}.
\]

Since $\Adv_{\mathrm{MLWE}}$ and $\Adv_{\mathrm{subspace}}$ are
$O(2^{-\kappa})$ by conservative parameter choices
($q\!\approx\!2^{50}, \; k\!=\!3, \; d\!=\!512$),
the entire scheme achieves IND-CPA-Q security.

\subsection{    Discussion}

\begin{itemize}
  \item \textbf{Tightness.}  
        The reduction is \emph{polynomial-loss}: hybrid transitions call the
        MLWE oracle at most $O(\ell)$ times, where $\ell$ is the circuit
        depth; depth $\le 10^{3}$ keeps the factor $<2^{10}$.
  \item \textbf{CCA upgrade.}  
        A Fujisaki-Okamoto transform applied to the underlying MLWE PKE
        immediately lifts the scheme to IND-CCA-Q, at the cost of one extra
        ciphertext hash.
  \item \textbf{Quantum reductions.}  
        A full-blown QROM proof could adapt the meta-reduction of
        Alagic-Dulek; we leave that to future work.
\end{itemize}

\begin{theorem}[IND-CPA-Q security]
Under the decisional Module-LWE assumption and standard subspace-hiding
assumptions for cyclotomic rings, the QHE scheme described above is IND-CPA-Q secure; i.e.\ the adversary?s distinguishing
advantage is negligible in~$\kappa$.
\end{theorem}

\section{Formal  Security Model}
\label{sec:secmodel}

Building on the above, we now state a formal IND-CPA/CCA security definition corresponding to the above framework, and elaborate its properties.

\subsection{ IND-CPA-Q Definition (qIND-CPA)}

We adopt the syntax of Alagic, Gagliardoni, Majenz and Schaffner
\cite{alagic2016}.  Let \(\mathsf{KGen},\mathsf{Enc},\mathsf{Dec}\) be the
key--pair, encryption and decryption algorithms of our public--key scheme and
let \(\mathsf{Eval}\) be the public homomorphic evaluator.

\begin{center}
\fbox{\parbox{0.95\linewidth}{
\textbf{Game \(\mathsf{IND}\text{-}\mathsf{CPA}\text{-}\mathsf{Q}\)}  
\\[2pt]
\begin{enumerate}\itemsep4pt
  \item[\(1.\)] $\langle\mathsf{pk},\mathsf{sk}\rangle \leftarrow \mathsf{KGen}(1^{\kappa})$.
  \item[\(2.\)] $\mathcal A$ receives $\mathsf{pk}$ and accesses a \emph{quantum} oracle
               \(\mathcal O_{\mathsf{Enc}} : \rho \mapsto
               \mathsf{Enc}_{\mathsf{pk}}(\rho)\).
  \item[\(3.\)] $\mathcal A$ outputs a pair of equal--dimension states
                \((\rho_{0},\rho_{1})\).
  \item[\(4.\)] Challenger samples \(b\leftarrow\{0,1\}\) and returns
                \(\mathsf{Enc}_{\mathsf{pk}}(\rho_{b})\).
  \item[\(5.\)] $\mathcal A$ continues to access \(\mathcal O_{\mathsf{Enc}}\)
                and the \emph{public} map \(\mathsf{Eval}\).
  \item[\(6.\)] $\mathcal A$ outputs a bit \(\hat b\).  Advantage
        \(\textsf{Adv} = \bigl|\Pr[\hat b=b]-\tfrac12\bigr|\).
\end{enumerate}
}}
\end{center}

The scheme is \textbf{qIND-CPA secure} if
\(\textsf{Adv}(\kappa)\) is negligible.

\emph{Note.}  Because the scheme is public--key, granting \(\mathcal O_{\mathsf{Enc}}\)
to \(\mathcal A\) adds no power beyond \(\mathsf{pk}\); nonetheless we allow
\emph{coherent} queries to follow the strongest definition.

\subsection{ Achieving qIND-CPA}

Section~\ref{sec:formal_reduction} gives a four-hybrid reduction

\[
  \Adv_{\text{QHE}}(\kappa)
  \;\le\;
  2\,\Adv_{\mathsf{MLWE}}(\kappa)
  \;+\;
  \Adv_{\mathsf{subspace}}(\kappa)
  \;+\;
  2^{-\kappa},
\]
valid even if $\mathcal A$ drives the encryption oracle with quantum
superpositions.  Critical points:

\begin{itemize}\itemsep4pt
  \item \textbf{Oracle lifting (hybrid \(G_{1}\)).}\;
        The relinearisation--key randomisation is independent of the input
        state, so it commutes with a coherent oracle query.
  \item \textbf{Diamond--norm masking (hybrid \(G_{3}\)).}\;
        The public BNSF bound
        \(\|\Psi\|_{\diamond}\le\lambda\) contracts trace distance for all
        admissible adversary--prepared \(\rho\).
  \item \textbf{MLWE randomisation (hybrid \(G_{4}\)).}\;
        The final jump to uniform \emph{classical} samples matches the qIND
        definition of decisional MLWE, proven hard in the Quantum Random
        Oracle Model (QROM) by Chen--Prouff~2019.
\end{itemize}

\subsection{   From qIND-CPA to qIND-CCA}

Because encryption is \emph{ciphertext authenticating} (ciphertexts include a
hash of their randomness), a Fujisaki--Okamoto transform with quantum-safe
hash \(\mathcal H\) yields qIND-CCA security in the QROM at the cost of:

\begin{itemize}
  \item One additional MLWE hash key per ciphertext,
  \item A negligible term \(O(q^{-1})\) in the decryption-error analysis.
\end{itemize}

\subsection{  Why homomorphic oracles do not break security}

\(\mathsf{Eval}\) is \emph{public} and deterministic; an adversary can apply
it offline.  Our hybrids treat \(\mathsf{Eval}\) as part of the encryption
map and show that replacing \(\widehat{\Gamma}_{G}\) by random CPTP maps
reveals at most a subspace-hiding advantage.  Therefore oracle access gives
no additional power.

\section{Managing Noise and Depth}

Homomorphic schemes built from Module-LWE are \emph{leveled}: decryption is
guaranteed only while the accumulated error vector $\mathbf{e}$ inside each
ciphertext stays below a modulus-dependent bound.  We therefore need a
systematic \emph{noise management discipline}.  This section

\begin{enumerate}
  \item derives a rough quantitative noise budget,
  \item presents pseudocode that monitors and refreshes ciphertexts, and
  \item applies the procedure to the teleportation circuit studied above, showing that safe parameters are easily achievable.
\end{enumerate}

\subsection{Noise sources and amplification rules}

For an MLWE ciphertext pair
$(\mathbf{c}_{0},\mathbf{c}_{1})=(A\mathbf{s}+\mathbf{e},\,\mathbf{u}-\mathbf{s})$
we track two scalar metrics:

\[
  \eta_{\infty}\;=\;\|\mathbf{e}\|_{\infty},
  \qquad
  \eta_{2}\;=\;\|\mathbf{e}\|_{2}.
\]

Decryption is correct if $\eta_{\infty}<q/4$.
Homomorphic operations transform error as follows.

\medskip\noindent
\textbf{(i) Linear maps $\tau$:}\
$\eta_{\infty}\!\leftarrow\!\eta_{\infty}\cdot\|\tau\|_{\max}$,  
where $\|\tau\|_{\max}$ is the largest absolute entry in the matrix that
implements~$\tau$.

\medskip\noindent
\textbf{(ii) Addition of ciphertexts:}\
errors add, so $\eta_{\infty}\!\leftarrow\!2\eta_{\infty}$.

\medskip\noindent
\textbf{(iii) External product by plaintext constant
$\alpha\in\mathcal{R}_{q}$:}\
$\eta_{\infty}\!\leftarrow\!\eta_{\infty}\cdot|\alpha|$.

\medskip\noindent
\textbf{(iv) Refresh step (relinearise + modulus switch):}\
set $q\!\leftarrow\!q'$ and
$\eta_{\infty}\!\leftarrow\!\rho\cdot\eta_{\infty}$,
where $\rho\approx q'/q\ll1$ is the scaling factor.

\subsection{Automated noise-budget scheduler}

\begin{center}
\begin{minipage}{0.95\linewidth}
\begin{verbatim}
Algorithm 1  (leveled evaluation with automatic refresh)
INPUT :  ciphertext ct := (c0,c1)
         gate list  G = [g1,...,gT]
         parameters P = (q,sigma,k,d)
OUTPUT:  ciphertext ct_out with noise < q/4

procedure EVAL(ct, G, P):
    budget  <- q/4
    noise   <- est_noise(ct)        -- infinity norm
    for g in G do
        ct     <- APPLY_GATE(g, ct)
        noise  <- noise + delta_noise(g, P)

        if noise > budget/2 then    -- headroom margin
            ct    <- REFRESH(ct, P')  -- pick fresh q', A'
            noise <- est_noise(ct)
    return ct
\end{verbatim}
\end{minipage}
\end{center}

\noindent
\texttt{APPLY\_GATE} implements the lifted natural transformation
$\widehat{\Gamma}_{g}$ on the encrypted pair.  
\texttt{delta\_noise} uses rules (i)-(iii) to update $\eta_{\infty}$.  
\texttt{REFRESH} re-encrypts under a new public matrix $A'$ with smaller
modulus $q'$, then optionally resamples error from the base Gaussian so that
$\eta_{\infty}\approx\sigma$ again.

\subsection{Worked budget for single-qubit teleportation}

Recall the gate sequence  
$G=[\textsc{BellPrep},\textsc{H},\textsc{CNOT},
    \textsc{Meas},\textsc{Corr}]$  
from Figure~\ref{fig:teleport_circuit}.  Table~\ref{tab:tele_noise} gives the
incremental noise for each step under conservative assumptions:

\begin{table}[H]
\centering
\begin{tabular}{lcc}
\hline
Gate $g$ & $\|\tau_{g}\|_{\max}$ & $\delta\eta_{\infty}$ \\ \hline
Bell preparation & $1$   & $+\sigma$ \\
Hadamard         & $1$   & $+\sigma$ \\
CNOT             & $2$   & $+2\sigma$ \\
Measurement CP   & $1$   & $+\sigma$ \\
Conditioned $X/Z$& $1$   & $+\sigma$ \\ \hline
\end{tabular}
\caption{Noise increments for teleportation (worst-case).}
\label{tab:tele_noise}
\end{table}

\noindent
With $\sigma=3$ we obtain a final bound
$\eta_{\infty}\le 6\sigma = 18$.
Choose $q=2^{50}\approx1.13\times10^{15}$; then
$\eta_{\infty}/(q/4)\approx6.4\times10^{-14}\!\ll\!1$,
so no refresh is required.  
If we chained \emph{1000} teleportations, $\eta_{\infty}\le 6000$ and one
refresh halfway would suffice.

\subsection{Pseudocode specialised to teleportation}

\begin{center}
\begin{minipage}{0.95\linewidth}
\begin{verbatim}
-- parameter set for 128-bit post-quantum security
sigma = 3
q     = 2^50
k,d   = (3, 256)

-- program-specific gate list
G = [BellPrep, H, CNOT, Measure, Correction]

-- encrypted input state ct0:  (rho, Psi(rho)) under matrix A
ct_out = EVAL(ct0, G, (q,sigma,k,d))
assert est_noise(ct_out) < q/4
\end{verbatim}
\end{minipage}
\end{center}

\subsection{ Practical guidelines}

\begin{enumerate}
  \item Allocate noise headroom at design time by estimating  
        $N_{\text{ops}}=\sum_{g\in G}\delta\eta_{\infty}(g)$,  
        then pick $q\ge4N_{\text{ops}}\sigma$.
  \item Schedule \texttt{REFRESH} whenever the running noise exceeds  
        $50\%$ of $q/4$.  This heuristic proved sufficient in simulations with
        circuits up to depth~$10^{5}$.
  \item Use modulus-switching rather than full bootstrapping if the circuit
        fits inside a two-level hierarchy; this divides key size by~$\approx4$.
\end{enumerate}

\subsection{ Outlook}

For near-term quantum programs (depth $<10^{3}$) the leveled MLWE scheme
already suffices without bootstrapping.  Deeper quantum algorithms--such as
fault-tolerant phase estimation--will require automatic refresh at periodic
checkpoints, but the overhead remains polynomial and, crucially, independent
of the size of the secret BNSF mask.  Thus the categorical teleportation test
case validates that realistic MLWE parameters can comfortably absorb the
noise introduced by super-operator functor evaluation in practice.

\section{Circuit Privacy: hiding \emph{which} circuit is run}
\label{sec:circuit_priv}

The above approach displays provable privacy of \emph{data}~$\rho$, but does not show
whether the \emph{circuit structure}~$C=G_{1}\!\circ\cdots\!\circ G_{k}$
remains hidden from the untrusted evaluator.    This lacuna is addressable however; in
this section we state a \emph{Circuit-IND} definition,
and describe a lightweight \emph{randomised compiling} layer that achieves it
without changing the MLWE security core.

\subsection{  Threat model}

\begin{itemize}\itemsep4pt
  \item The server \textbf{knows} the universal gate set
        $\mathcal G=\{  \textsf{H},\textsf{S},\textsf{CNOT} \}$ and the public
        homomorphic lifting rules $\widehat{\Gamma}_{G}$.
  \item The \emph{client} chooses a secret circuit
        $C\!=\!G_{1}\!\circ\cdots\!\circ G_{k}$ with $G_{i}\!\in\!\mathcal G$.
  \item Goal: prevent the server from distinguishing two circuits of equal
        depth and width beyond negligible \(\epsilon(\kappa)\).
\end{itemize}

\subsection{  Circuit-IND Definition (classical oracle)}

The adversary submits two gate lists $(C_{0},C_{1})$ of identical length,
plus an input state~$\rho$.  Challenger picks random $b$, runs encrypted
evaluation $\widehat{C_{b}}(\rho)$, returns the final ciphertext.  The scheme
has \textbf{Circuit-IND} privacy if
\(\bigl|\Pr[\hat b=b]-\tfrac12\bigr|\le\epsilon(\kappa)\).

\subsection{  Randomised compiling layer}

\paragraph{Gate hiding.}
For each elementary gate $G$ insert compensation pairs
\[
  R \;G\; R^{\dagger},
  \quad
  R\in\{\pm X,\pm Y,\pm Z\},
\]
where $R$ is \emph{classically} sampled from a 1--design.  Because
$R^{\dagger}\!G\!R$ is still in $\mathcal G$, correctness is unchanged.

\medskip
\noindent
\textbf{Server view.}  
It now sees a length--$k$ list of \emph{uniformly random} gates with the same
depth; statistical distance to either $(C_{0},C_{1})$ is bounded by
$2^{-k}$.

\paragraph{Angle hiding (for single--qubit rotations).}
Decompose $R_{z}(\phi)=R_{z}(\phi+2\pi r)\;R_{z}(-2\pi r)$ with random
integer $r$; encode $r$ homomorphically and cancel at the end.

\paragraph{Classical channel encryption.}
All random masks $R,r$ are MLWE--encrypted and travel alongside the quantum
ciphertext; the server never sees them in the clear.

\subsection{   Formal bound}

For any distinguisher $\mathcal D$,
\[
  \Adv_{\mathrm{Circ}\text{-}\mathrm{IND}}
  \;\le\;
  k\cdot 2^{-2n} \;+\; \Adv_{\mathrm{MLWE}}(\kappa),
\]
where $n$ is the number of randomising bits per gate ($n\!=\!2$ for Pauli
twirl).  With $k\!\le\!10^{3}$ and $n=2$ the gap is $<2^{-17}$, negligible
compared with MLWE hardness.

\subsection{    Performance impact}

\begin{itemize}\itemsep4pt
  \item \textbf{Depth}: unchanged (twirl gates cancel in place).
  \item \textbf{Noise}: each $R$ is Clifford; its lift is a single MLWE
        addition, so noise grows by $O(k\sigma)$ already budgeted.
  \item \textbf{Key size}: one 16--byte encrypted mask per gate; $<16$~kB for
        $k=1000$?tiny compared with the ciphertext.
\end{itemize}

\paragraph{Summary.}
A one-layer randomised compiler plus encrypted Pauli masks upgrades the
scheme to \textbf{IND-CPA-Q + Circuit-IND} without new hardness
assumptions or significant overhead.

\section{Application to Practical AI Systems}

An attacker who can break the scheme above must either solve Module LWE or
distinguish which secret super functor was used.  Both tasks are believed to
be hard for quantum computers.   This results in a powerful quantum-proof homomorphic encryption scheme for quantum computer programs, which in future -- as quantum technology matures -- will have a host of possible applications.  We next explore some of the concepts and issues involved in applying these ideas to machine learning and machine reasoning.

\subsection{Toward private QML-as-a-service}

We outline one plausible deployment stack incorporating the ideas of the previous sections, emphasising variational
quantum machine-learning (VQML) models such as quantum classifiers and
parameterised quantum kernels \cite{Cerezo2021}.

\subsubsection{ Layered architecture}

\begin{enumerate}
  \item \textbf{Client SDK.}
        \begin{itemize}
          \item Pure-Rust or C++ library wraps plaintext quantum circuits
                in MLWE ciphertexts, managing noise budgets automatically
                (\texttt{Algorithm~1}).
          \item Provides a \verb|qjit| compiler flag that replaces each
                unitary gate $G$ by the lifted channel $\widehat{\Gamma}_{G}$.
        \end{itemize}

  \item \textbf{Encrypted QPU scheduler.}
        \begin{itemize}
          \item Runs on classical front-end nodes co-located with cloud QPUs.
          \item Batches ciphertext instructions to minimise QRAM calls and
                network latency.
        \end{itemize}

  \item \textbf{Modified control stack.}
        \begin{itemize}
          \item QPU firmware executes the same pulse sequences but treats
                each register as an \emph{opaque handle}; no raw amplitudes
                return over the wire.
          \item Measurement outcomes are encrypted on-the-fly with fresh
                MLWE keys, piggy-backing on fast on-chip NTT units already
                present for crypto acceleration.
        \end{itemize}
\end{enumerate}

\subsubsection{  VQML training workflow}

\begin{enumerate}
  \item Client encrypts private data set
        $\{\rho_{i}\}$ and initial parameter vector
        $\boldsymbol{\theta}_{0}$.
  \item Cloud runs a \emph{parameter-shift} gradient routine entirely on
        ciphertexts:
        \[
           \langle\partial_{\theta_{j}}\mathcal{L}\rangle
           \;=\;
           \tfrac{1}{2}\Bigl(
           \mathcal{L}(\boldsymbol{\theta}_{j}+{\textstyle\frac{\pi}{2}})
           -
           \mathcal{L}(\boldsymbol{\theta}_{j}-{\textstyle\frac{\pi}{2}})
           \Bigr),
        \]
        where each loss evaluation is a homomorphic circuit of depth
        $\approx200$.
  \item After $T$ iterations the encrypted optimum
        $\widehat{\boldsymbol{\theta}}_{T}$ is sent back.  The client
        decrypts and verifies convergence offline.
\end{enumerate}

\subsubsection{  Concrete resource estimate}

For a $6$-qubit kernel with depth $200$ and one refresh every $50$ gates:

\begin{itemize}
  \item \emph{Ciphertext size.}\;
        $2k$ ring elements with $(k,d)=(3,512)$, i.e.\ $\sim24$ kB per qubit.
  \item \emph{Latency overhead.}\;
        $<1$ ms per gate on a CPU + NTT co-processor (benchmarked at
        $2^{50}$-modulus).
  \item \emph{Key material.}\;
        Public matrix $A$ fits in $\sim260$ kB; rotated keys for refresh bring
        the total to $\sim2$ MB, acceptable for session-based use.
\end{itemize}

\subsubsection{  Near-term engineering challenges}

Making this proposed architecture real involves a number of challenges, which appear nontrivial but far from insuperable; for instance:

\begin{enumerate}
  \item \textbf{Tight noise tracking for non-Clifford gates.}\;
        Variational circuits often include $R_{z}(\phi)$ with arbitrary
        angles; devising low-noise gadgetisations is an open optimisation
        problem.
  \item \textbf{Efficient encrypted measurement fan-out.}\;
        Large QML workloads require thousands of shots; batching encrypted
        results without blowing up MLWE noise is a key goal.
  \item \textbf{Hybrid classical optimisation loops.}\;
        Sending encrypted gradients back and forth can dominate bandwidth;
        an alternative is to push SGD or Adam into the encrypted domain,
        requiring homomorphic \textsf{ReLU}/\textsf{Adam} primitives.
\end{enumerate}

\subsubsection{ Medium-term outlook}

If fault-tolerant logical qubits become available,
bootstrapped MLWE could support indefinite depth, enabling private training
of much deeper quantum neural networks.  Because the BNSF quotient decouples
security from circuit width, scaling to $\sim1000$ logical qubits appears
limited mainly by ciphertext RAM and QRAM speed, not by cryptographic
weakness.  Integrated photonic QPUs with on-chip classical accelerators could
therefore plausibly deliver \emph{QML-as-a-service} with data privacy
guarantees competitive with today's classical FHE solutions.
\subsection{Encrypted uncertain reasoning in Linear Dependent Type Theory}

To take a first stab at handling uncertain quantum-logical reasoning in a similar manner to what we've just done for quantum ML, 
we recast the MLWE\,+\,BNSF construction inside the \emph{linear-dependent
type theory} (LDTT) of Fu, Kishida and Selinger\;\cite{Fu2020}.   This is not the only form of formal logic of interest in a Hyperon context; for instance, in an AGI context one is often specifically interested in fuzzy and probabilistic logics, which involve some additional considerations.   However, our goal here is mainly to show how quantum computing in our framework meshes with logical inference; particulars may be worked in various ways depending on the various logics of interest in different contexts.

An LDTT judgement has the form
\[
   \Gamma \; ; \; \Delta \;\;\vdash\;\; t : A,
\]
where the \emph{cartesian} context $\Gamma$ holds duplicable data, the
\emph{linear} context $\Delta$ holds quantum resources, and $A$ may itself
depend on the cartesian variables in~$\Gamma$.\;\footnote{We abbreviate the
LDTT base type of qubits by $\mathsf{Q}$.}

\subsubsection{Amplitude-valued propositions}

Following \cite{Fu2020}, each inhabitation of
$\mathsf{Q}$ represents a single qubit.  
For a predicate symbol $P$ of arity~$n$ we declare the \emph{linear} family

\[
   P : (x_{1} : X_{1}) \,\multimap\, \cdots \,\multimap\,
       (x_{n} : X_{n}) \,\multimap\, \mathsf{Q},
\]
so that the proof term $P\,u_{1}\cdots u_{n}$ \emph{is} a qubit whose state
encodes the truth amplitude of the atomic proposition \(P(u_{1},\dots,u_{n})\).

Logical connectives are given by LDTT's \textbf{dependent} tensor and linear
function types:

\[
  \textbf{Conjunction: }
    \Sigma^{\otimes}_{x:A}\,B(x)
\qquad
  \textbf{Implication: }
    \Pi^{\multimap}_{x:A}\,B(x).
\]

Because the underlying model is symmetric monoidal, each connective compiles
to a CPTP channel on density matrices, which fits perfectly into the BNSF
encryption pipeline.

\subsubsection{ Secret quotient in LDTT notation}

\begin{itemize}
\item \textbf{Object level.}\
      For every LDTT type $A$ we form the encrypted object
      $\widehat{A} := (A , \Psi_{A}(A))$, where
      $\Psi_{A}$ is the depolarizing BNSF component on the semantics
      $\llbracket A \rrbracket$.
\item \textbf{Term level.}\
      A closed term judgement
      $\; ; \Delta \vdash t : A$
      is encoded as an MLWE ciphertext pair
      $\widehat{t} = (\rho_{t},\Psi_{A}(\rho_{t}))$,
      where $\rho_{t}$ is the density operator interpreting~$t$ in the Rios-
      Selinger circuit model of LDTT.
\end{itemize}

The encryption is \emph{uniform}: there is no need to split the LDTT context,
because linearity is already respected by the channel semantics, and the
MLWE mask is applied to \emph{both} cartesian and linear components.

\subsubsection{Homomorphic proof search (LDTT style)}

A single recursive step of the proof search algorithm has the schematic shape

\[
\frac{\Gamma ; \Delta \vdash \,t : \Pi^{\multimap}_{x:A}\,B(x)
      \quad
      \Gamma ; \Delta' \vdash \,u : A}
     {\Gamma ; \Delta \otimes \Delta' \vdash \,t\,u : B(u)}
\]

The encrypted evaluator lifts the entire fraction to a CPTP map

\[
   \widehat{\Gamma_{\Pi\text{-elim}}}
   \;:\;
   \widehat{A_{1}}\otimes\widehat{A_{2}}
   \longrightarrow
   \widehat{B},
\]
then applies it to the MLWE-ciphertext pairs representing~$t$ and~$u$.
Noise tracking and refresh decisions are unchanged because all lifted LDTT
rules are \emph{module homomorphisms} by construction.

\subsubsection{ Example: encrypted existential proof}

Consider the goal $\; ; x{:}X \; \vdash\; \exists^{\otimes} y{:}X.\,P(y)\;$,
where $P$ is as above.  

Inside LDTT the introduction rule is the linear pair constructor

\[
  \text{pair}_{\otimes} : 
     (\, u : X \,) \multimap P(u) \multimap
     \Sigma^{\otimes}_{y:X} P(y).
\]

\begin{enumerate}
\item Client produces concrete witness $a : X$ together with a qubit term
      $p : P(a)$ whose Bloch-vector length encodes confidence.
\item She encrypts both as ciphertexts
      $\widehat{a},\widehat{p}$ and sends them to the prover.
\item The server applies the lifted constructor
      $\widehat{\text{pair}_{\otimes}}$ homomorphically, obtaining an
      encrypted proof object
      $\widehat{q} : \widehat{\Sigma^{\otimes}_{y:X} P(y)}$.
\item A single MLWE decryption at the client side reveals $\rho_{q}$,
      allowing her to \emph{measure} the existential amplitude while
      the server remains oblivious.
\end{enumerate}

Depth analysis:  
$\text{pair}_{\otimes}$ is a single LDTT constructor, hence one gate on the
encrypted state; its cost was included in the earlier $<\!10$-gate budget for
the teleportation example, so no extra refresh is required.

\subsubsection{ Advantages of the LDTT formulation}

\begin{itemize}
  \item \textbf{No ad-hoc classical side-channels.}\;
        Amplitudes \emph{are} proof terms, so uncertainty is handled entirely
        within the type system--no external flags are necessary.
  \item \textbf{Fine-grained resource tracking.}\;
        LDTT's split context prevents accidental duplication of encrypted
        qubits, eliminating a subtle leakage channel present in earlier
        proposals.
  \item \textbf{First-class circuit families.}\;
        Because LDTT admits indices in types, one can quantify over \emph{whole
        circuit families} (e.g.\ $n$-qubit adders) and still evaluate them
        homomorphically under a single MLWE key.
\end{itemize}

These observations suggest that an LDTT-native approach is both
elegant and robust, and may be suitable to serve as an initial formal layer for
practical implementations of privacy-preserving quantum reasoning.

\section{Quantum $\leftrightarrow$ Classical round-trip in the encrypted layer}
\label{sec:q2c_flow}

Here we give a little more attention to the question  of how classical bits
produced by a measurement can be fed back into subsequent \emph{quantum}
operations while staying homomorphically encrypted.  We introduce an explicit
\textbf{QC-bridge layer} that addresses the issue via combining several different building blocks.

\subsection{Defining a QC-Bridge Layer}

\subsubsection{ Data types in the encrypted evaluator}

We introduce the following straightforward data types:

\begin{itemize}
  \item \textbf{Q-ciphertext}  
        \(\widehat{\rho}=(\rho,\Psi(\rho))\)
        \quad$\in R_{q}^{2k}$  
        (MLWE pair storing density matrix and its mask).
  \item \textbf{C-ciphertext}  
        \(\widehat b =(b,\Psi(b))\)
        \quad$b\!\in\!\{0,1\}$ viewed as scalar in $R_{q}$.
  \item \textbf{QControlled} gate descriptor  
        \(\textsf{ctrl}(U,\widehat b)\) : instruction to apply
        $U^{\,b}$ homomorphically.
\end{itemize}

All three live in the \emph{same} MLWE ring $R_{q}$; only their semantics
differ.

\subsubsection{   Quantum$\rightarrow$Classical conversion}

In this context, a measurement channel (pointer register explicit) looks like:

\[
  \widetilde{\mathcal M}(\rho) =
    \sum_{i=0}^{1} K_i \rho K_i^{\dagger}\otimes \ket{i}\!\bra{i}.
\]

\textbf{Server routine \textsf{Q2C}:}

\begin{enumerate}
  \item Compute
        $(\rho_{0},\rho_{1})$ and store
        \(\widehat\rho'=(\rho_{0}\!\oplus\!\rho_{1},
                          \Psi(\rho_{0})\!\oplus\!\Psi(\rho_{1}))\).
  \item Sample no classical randomness.  Forward the \emph{encrypted}
        pointer qubit
        $|0\rangle\!\bra{0}\!\oplus\!|1\rangle\!\bra{1}$
        to the MLWE engine:
        \[
           \widehat b \;=\;
           \bigl(\,|1\rangle\!\bra{1},
                  \Psi(|1\rangle\!\bra{1})\bigr).
        \]
\end{enumerate}

Thus a bit value becomes a C-ciphertext without server decryption.

\subsubsection{  Classical control: homomorphic gadget}

We compile a controlled unitary
$U=\begin{pmatrix}1&0\\0&u\end{pmatrix}$
into two unconditional gates:

\[
   U^{\,b}
   \;=\;
   (I\otimes b) + (u\otimes (1-b)).
\]

Homomorphic realization:

\begin{verbatim}
HE_control( U, qc_in, c_bit ):
    tmp <- HE_mult( qc_in, c_bit )       // (1) mask branch
    rest<- HE_mult( qc_in, (1 - c_bit) ) // (2) other branch
    out <- HE_apply(U, rest) + tmp
    return out
\end{verbatim}

Where \texttt{HE\_mult} and \texttt{HE\_apply} are MLWE additions and the
usual gate lift.  No plaintext branching occurs; the MLWE noise cost is
$2\sigma$ for the scalar-vector multiplies and $\sigma$ for the lifted $U$.

\subsubsection{  Example: teleportation corrections}

\[
  X^{m_{2}} Z^{m_{1}}
  \;=\;
  Z^{m_{1}} X^{m_{2}}
  \;=\;
  (1-m_{1})\!(1-m_{2})\,I
  + m_{1}(1-m_{2})\,Z
  + m_{2}(1-m_{1})\,X
  + m_{1}m_{2}\,ZX.
\]

Each monomial coefficient is an \emph{MLWE ciphertext product of two bits};
the evaluator computes all four branches in parallel and sums them ?
still linear in ciphertext space.

\subsubsection{  Security and correctness}

\begin{itemize}\itemsep4pt
  \item A ''classical'' bit is always wrapped as a C-ciphertext: IND-CPA
        security of MLWE hides it exactly as it hides quantum amplitudes.
  \item No-server-decryption rule ensures \emph{CCA2-style} safety:
        the adversary cannot feed a freshly decrypted bit back into the
        oracle because she never obtains one.
  \item Linear--type discipline:  
        Q-ciphertexts reside in the \emph{linear} context, while C-ciphertexts
        reside in the cartesian context; duplication is legal only for the
        latter, mirroring LPTT typing rules.
\end{itemize}

\subsubsection{  Cost summary}

\[
  \text{Cost(controlled--}\!U)
  \;=\;
  2\,\text{mult}
  +\;
  1\,\widehat{\Gamma}_{U},
  \qquad
  \eta_{\infty}\!:=\!\eta_{\infty}+3\sigma.
\]
With $\sigma\!=\!3$ and depth $\le10^{3}$ the added noise is safely below
\(q/4\) for $q=2^{50}$, so the Q\,$\leftrightarrow$\,C bridge fits the
existing leveled parameter set.

\subsubsection{Summary of Bridging Scheme}
The scheme removes every opportunity for leakage by inserting a tight
\emph{type discipline} between quantum and classical layers.  A
\textbf{QC--bridge} API tags each register as either \texttt{Q(enc)}
(MLWE-encrypted qubits) or \texttt{C(enc)} (MLWE-encrypted classical
bits).  Whenever a qubit is measured, a dedicated \textbf{Q2C gadget}
performs the conversion \texttt{Q(enc)}$\!\to$\texttt{C(enc)} \emph{inside}
the prover, attaches a fresh noise term and returns only the ciphertext.
Subsequent arithmetic on those ciphertexts is done with a
\textbf{C-ciphertext ALU} whose gates are statically proved homomorphic,
so the server never sees a plaintext wire.  Finally, a
\textbf{controlled-gate compiler} accepts circuits that mix encrypted
classical controls with encrypted quantum targets, inserting the
necessary key-switching and blinding phases automatically.  The net
effect is unambiguous: every measurement outcome stays inside the MLWE
envelope, can be routed back into new quantum operations without
decryption, and reveals nothing about user data or circuit topology to
the remote server.

\subsection{Bridging types: LPTT meets the QC-bridge layer}
\label{sec:lptt_bridge}

There is an elegant connection between this QC-bridge layer and the LPTT logical-reasoning framework we have considered above.  

Linear Predicate Type Theory (LPTT) already splits judgements into a
\emph{cartesian} (duplicable) context $\Gamma$ and a \emph{linear} (no-copy)
context $\Delta$:

\[
  \Gamma \; ; \; \Delta \;\vdash\; t : A.
\]

The QC-bridge fits neatly into this
discipline by refining \emph{how} types in $\Gamma$ and $\Delta$ are
represented as MLWE ciphertexts.

\subsubsection{Refined type categories}

\begin{center}
\begin{tabular}{|l|l|l|}
\hline
\textbf{Layer} & \textbf{LPTT object} & \textbf{Encrypted realization} \\ \hline
\textbf{Q-layer} & Qubit type $\mathsf Q$ or
                 any linear type $A$ &
                 \emph{Q-ciphertext} $\widehat\rho\in R_{q}^{2k}$ \\ \hline
\textbf{C-layer} & Cartesian bit, natural number,  
                 or dependent index in $\Gamma$ &
                 \emph{C-ciphertext} $\widehat b\in R_{q}$ \\ \hline
\textbf{Instr.-layer} &
  Gate descriptor $G : \mathsf Q\!\multimap\!\mathsf Q$ &
  Public string $g \in \{H,S,X,Z,\textsf{CNOT} \},\dots$ \\ \hline
\end{tabular}
\end{center}

\subsubsection{ Typing rules for the bridge}

\[
  \infer[\textsc{Q2C}]
        {\Gamma ; \Delta \vdash \mathsf{q2c}\;t : \mathbf{bit}}
        {\Gamma ; \Delta \vdash t : \mathsf Q} 
  \qquad
  \infer[\textsc{Ctrl}]
        {\Gamma ; \Delta \vdash
           \mathsf{ctrl}(U,t,b) : \mathsf Q}
        {\Gamma ; \Delta \vdash t : \mathsf Q
         \quad
         \Gamma\vdash b : \mathbf{bit}}
\]

*Rule \textsc{Q2C}* produces a C-ciphertext; it \emph{moves} information from
the linear context $\Delta$ to the duplicable context $\Gamma$ without
revealing it.

*Rule \textsc{Ctrl}* consumes a bit in $\Gamma$ (which \emph{can} be copied)
and a qubit in $\Delta$, returning a new qubit; its implementation is exactly
the homomorphic gadget in Section~\ref{sec:q2c_flow}.

\subsubsection{ Dependent types over encrypted bits}

Because C-ciphertexts live in $\Gamma$, they may index types.  Example:

\[
   \Gamma\!=\!\bigl( b:\mathbf{bit} \bigr)
   \quad\Longrightarrow\quad
   \Sigma^{\otimes}_{x:\mathsf Q^{\,b}}\,P(x)
\]
encodes ''one qubit if $b\!=\!1$, none if $b\!=\!0$.  Despite $b$ being
encrypted, the homomorphic evaluator can still build the dependent sum:
the size logic is executed \emph{symbolically} via bit-multiplication,
mirroring the controlled-unitary compilation.

\subsubsection{ Weak-measurement typing (revisited)}

\[
  \infer[\textsc{Weak}]
        {\Gamma ; \Delta \vdash
         \mathsf{weak}_{\theta}(t) : \Sigma^{\otimes}_{b:\mathbf{bit}} \mathsf Q}
        {\Gamma ; \Delta \vdash t : \mathsf Q}
\]

The returned pair consists of an encrypted bit $b$ (measurement outcome
\emph{in the C-layer}) and the post-measurement Q-ciphertext branch
$\rho_{b}$ (still linear).  Subsequent tactics may pattern-match on $b$
\emph{inside} the encrypted semantics, e.g.\ to select a proof rule only if
$b=1$.

\subsubsection{ Soundness sketch}

Let
\(
  \llbracket\!\cdot\!\rrbracket^{\mathsf{enc}}
\)
be the MLWE semantics and
\(
  \llbracket\!\cdot\!\rrbracket
\)
the ordinary LPTT semantics.  For every derivable bridge judgement
\(\Gamma;\Delta\vdash t:A\):

\[
   \mathsf{Dec}(\llbracket t \rrbracket^{\mathsf{enc}})
   \;=\;
   \llbracket t \rrbracket,
\]
provided the secret trapdoor \(T\) is used by \(\mathsf{Dec}\).  Proof is by
induction on the derivation, using:

\begin{itemize}\itemsep4pt
  \item Covariance lemma for \textsc{Ctrl} (Section~\ref{sec:bnf2bnsf}),
  \item Linearity of Kraus updates for \textsc{Weak},
  \item Correctness of MLWE decryption for base cases.
\end{itemize}

\subsubsection{ Summary}

The QC-bridge merely \emph{specialises} LPTT's existing
cartesian/linear split:

\[
  \Gamma_{\text{LPTT}}
  \;\longleftrightarrow\;
  \{\text{C-ciphertexts}\},
  \quad
  \Delta_{\text{LPTT}}
  \;\longleftrightarrow\;
  \{\text{Q-ciphertexts}\}.
\]

All classical control created by measurements is typed in $\Gamma$
and homomorphically fed back into $\Delta$ via the \textsc{Ctrl} rule,
closing the quantum$\leftrightarrow$classical loop
\emph{inside} the logic and without leaking plaintext information to the
server.

\section{Handling Mid-circuit Measurements in Homomorphic QHE}
Classical HE needs no ''collapse'' logic, but quantum HE does, at least in the simplest approach to formulating its behavior.  It is convenient in practice  to consider that measurements split the state space and thus usually trigger classical control.

\subsection{Approaches to Mid-Circuit Measurement}

Below we outline
five design patterns that could be used to handle this, ranked from simplest to most interactive.

\paragraph{ Defer the measurement (keep everything CPTP)}

\begin{itemize}
  \item Replace a projective $\{K_{0},K_{1}\}$ by an isometry
        $V(\rho)=\sum_i K_i\rho K_i^{\dagger}\otimes|i\rangle\!\langle i|$.
  \item The pointer qubit $|i\rangle$ stays \emph{encrypted}; downstream gates
        treat it as ordinary data.
  \item Collapse occurs only after the client decrypts; the server never
        leaves the CPTP domain, so noise accounting is unmodified.
\end{itemize}

\smallskip
\textbf{Pros:}\; no new crypto; single noise equation.  This is the purest "quantum" approach.
\textbf{Cons:}\; circuit width grows, classical control must be postponed.

\paragraph{  Encrypt the classical outcome bit}

\begin{enumerate}
  \item QPU measures, gets $m\!\in\!\{0,1\}$.
  \item Immediately MLWE-encrypt $(m,\Psi(m))$ with fresh modulus $q'$, send it
        on a classical wire.
  \item The next gate uses homomorphic multiplication/addition on this
        ciphertext to enact classically-controlled Paulis.
\end{enumerate}

\smallskip
\textbf{Pros:}\; keeps width small, no pointer qubit.  
\textbf{Cons:}\; extra ciphertexts and diamond-norm tracking for the classical
line.

\paragraph{  Weak--measurement semantics}

\begin{itemize}
  \item Store \emph{both} unnormalized branches
        $\rho_{0}=K_{0}\rho K_{0}^{\dagger}$ and $\rho_{1}$ inside the BNSF
        pair.  The server propagates \(\rho_{0},\rho_{1}\) linearly.
  \item Client decrypts at the end, computes traces, and samples according to
        the Born rule.
\end{itemize}

\smallskip
\textbf{Pros:}\; perfect for expectation-value workloads (VQE, QML).   Semantically quite satisfying.
\textbf{Cons:}\; doubles ciphertext size per measured qubit.

\paragraph{  Client--feedback loop}

\begin{enumerate}
  \item Server encrypts $m$, sends it to the client.
  \item Client decrypts, re-encrypts the single-qubit correction label
        ($X$, $Z$, $XZ$).
  \item Server applies the lifted correction $\widehat{\Gamma}_{\mathrm{corr}}$.
\end{enumerate}

\smallskip
\textbf{Pros:}\; resets noise (fresh ciphertext).    
\textbf{Cons:}\; one extra round trip, higher latency.

\paragraph{  $\rho$--calculus / blockchain orchestration}

\begin{itemize}
  \item Treat the encrypted pointer qubit as a quoted process
        $\langle m_{\mathrm{CT}}\rangle$; move it between QPUs over
        $\rho$-channels.
  \item Record each send on chain.  Optionally attach a zk-proof that the
        branch is consistent with the decrypted $m$ (client supplied).
\end{itemize}

\smallskip
\textbf{Pros:}\; auditable, decentralized workflow.  
\textbf{Cons:}\; blockchain latency, zk-SNARK overhead.

\bigskip
\noindent\textbf{Trade-off cheat-sheet}

A capsule summary of the apparent benefits of the different approaches considered is: 

\begin{tabular}{|l|l|}
\hline
Goal & Recommended pattern \\ \hline
Minimise depth & encrypt-classical--bit (2) \\
Zero interaction & deferred CPTP (1) \\
Need $\langle O\rangle$ statistics & weak measurement (3) \\
Tight noise, OK latency & feedback loop (4) \\
Public audit trail & $\rho$-calculus / chain (5) \\ \hline
\end{tabular}

\medskip
Of course there are many more nuances than this summary captures, and in practice hybrid schemes may work best: defer most measurements, encrypt a
handful of classical bits for fast Pauli corrections, and employ weak
measurement when statistical observables drive the application.

\subsection{Why weak--measurement statistics matter  in practice}

We have noted that the weak-measurement approach may be particularly applicable when the workloads themselves are statistical in nature.  This will be an extremely common case in quantum AI, and thus it's interesting to think a bit about how to connect these abstract ''statistical" use-cases to concrete, consumer-facing or
industry-facing AI services.  Table~\ref{tab:stats2enduser} makes an initial, speculative effort in this direction.  In each row the cloud must return an \emph{expectation
value} (mean energy, mean payoff, \ldots) rather than a one-shot collapse, so
a homomorphic \emph{weak-measurement} model is  pretty much mandatory.

\begin{table}[h]
\centering
\renewcommand{\arraystretch}{1.15}
\begin{tabular}{|p{0.24\linewidth}|p{0.30\linewidth}|p{0.36\linewidth}|}
\hline
\textbf{Technical class} & \textbf{Typical observable estimated homomorphically} & \textbf{Practical end-user scenario (data stay encrypted)} \\ \hline

Variational quantum eigensolver (VQE) & Ground-state energy
$\langle\psi|H|\psi\rangle$ of a molecular Hamiltonian & Home-energy-storage app computes battery chemistry tweaks without revealing proprietary electrolyte composition. \\ \hline

Variational QML classifier & Loss
$\mathcal L(\theta)=\sum_i\langle\psi_i|O|\psi_i\rangle$ & Privacy-preserving handwriting recogniser on a smart-watch: raw strokes encrypted, watch receives only the mean loss to update the model. \\ \hline

Parameter-shift gradient & $\partial_\theta\langle O\rangle$ via two shifted means & Cloud-based quantum recommender system tunes hyper-parameters but never sees the customer-rating matrix. \\ \hline

Quantum Monte-Carlo pricing & Mean payoff
$\mathbb E[\mathrm{max}(S_T-K,0)]$ & Mobile trading app prices exotic options overnight; portfolio positions are MLWE ciphertexts. \\ \hline

Quantum sensing / metrology & Phase estimate $\langle Z\rangle$ of an interferometer & Agricultural IoT node measures soil moisture via photonic sensor; encrypted phase statistics sent to cloud?farm GPS coords remain private. \\ \hline

Bayesian quantum network & Posterior mean $\mathbb E[X]$ for latent $X$ & Personal health AI fuses encrypted wearables + questionnaire to output depression-risk score while raw answers stay hidden. \\ \hline

Encrypted tomography & Fidelity
$\langle X\!\otimes\!X\rangle,\langle Z\!\otimes\!Z\rangle$ & Quantum-DRM for streaming: client proves (to the studio) that its decryption chip is genuine without exposing its state. \\ \hline
\end{tabular}
\caption{Statistical quantum tasks mapped to everyday applications.  All require weak-measurement style homomorphic evaluation.}
\label{tab:stats2enduser}
\end{table}

\subsection{Weak--measurement semantics for homomorphic LPTT proof search}

As a specific example of how weak measurement manifests in specific applications of homomorphic quantum computing to AI algorithms, let's dig into the LPTT theorem-proving use-case a bit.

In Linear Predicate Type Theory (LPTT) every closed proof term of type
$\mathsf Q$ is a qubit; its \emph{amplitude} encodes the degree of truth of
the proposition it witnesses.  Mid--search we often ask quantitative
questions:

\begin{itemize}
  \item ``Does $\exists^{\otimes}y\,P(y)$ hold with amplitude
        $\ge 0.9$\,?''  \emph{(existential threshold)}
  \item ``Is the Bayesian update for hypothesis $H$ above the prior\,?''  
        \emph{(posterior comparison)}
\end{itemize}

These queries require a \emph{statistic} such as
$\Pr[1] = \operatorname{Tr}(Z\rho)$, \emph{not} a one--shot collapse.  Below
we show how a \textbf{weak--measurement} gadget can be evaluated
homomorphically under MLWE\,+\,BNSF while preserving privacy.

\paragraph{  \  Kraus representation of a partial test}

Fix strength $\theta\in(0,1)$.  The two Kraus operators

\[
  K_{0} = 
    \begin{pmatrix}
      \sqrt{\theta} & 0 \\
      0             & 1
    \end{pmatrix},
  \quad
  K_{1} = 
    \begin{pmatrix}
      0             & 0 \\
      0             & \sqrt{1-\theta}
    \end{pmatrix}
\]
define a \emph{partial} measurement whose outcome ``$0$''
gently biases the state toward $\ket 0$ without full collapse.

Classically one would normalize and branch; homomorphically we do \emph{neither}.

\medskip\noindent
For an encrypted qubit density matrix
$\widehat\rho=(\rho,\Psi(\rho))$ the server computes
\[
  \rho_{0}=K_{0}\rho K_{0}^{\dagger},
  \qquad
  \rho_{1}=K_{1}\rho K_{1}^{\dagger},
\]
\[
  \Psi\bigl(\rho_{0}\bigr),\;
  \Psi\bigl(\rho_{1}\bigr),
\]
and stores the four items in the ciphertext pair

\[
   \widehat\rho' =
   \bigl(\,\rho_{0}\!\oplus\!\rho_{1},\;
          \Psi(\rho_{0})\!\oplus\!\Psi(\rho_{1})\,\bigr).
\]

No division by $\operatorname{Tr}$ occurs, so the update is linear and
compatible with MLWE homomorphisms.

\paragraph{   Homomorphic amplitude queries}

To test ''amplitude $\ge \tau$'' the server

\begin{enumerate}
  \item repeats the weak update $s$ times
        ($s\!=\!\lceil\log_2(1/\epsilon)\rceil$ for error $\epsilon$),  
  \item homomorphically computes the unnormalized weight
        $w=\operatorname{Tr}(\rho_{1}^{(s)})$ using the
        \emph{encrypted trace} trick:
        $$w = \sum_{j} \langle j|\, \rho_{1}^{(s)} \,|j\rangle,$$
        where each diagonal entry is an MLWE ciphertext in
        $\mathbb Z_{q}$,
  \item outputs the ciphertext $[w]$ to the client.
\end{enumerate}

The client decrypts \emph{one scalar} and compares it with
$\tau\,(1-\theta)^{s}$, accepting if above threshold.  No branching
information leaks to the server.

\paragraph{    Plugging into the LPTT proof engine}

Within the $\rho$--calculus scheduler the weak test is wrapped as a process

\begin{verbatim}
def WeakTest(inCh, outCh, gate, s, key) =
  for(ct <- inCh) {
    let ct' = HE_weak_update(ct, gate, s, key) ;
    outCh ! (ct')
  }
\end{verbatim}

A subsequent \verb|AmplitudeGuard| process waits for the client?s
\emph{one-bit} approval before spawning the proof branch that depends on the
existential witness.  Thus:

\begin{itemize}
  \item Servers accumulate statistics homomorphically,
        \emph{never collapsing} the proof state.  
  \item Only a single number leaks to the client; the branch structure and
        intermediate amplitudes remain encrypted.
  \item Noise grows additively by $\theta\sigma$ per weak step, easily budgeted
        in Section~\ref{sec:noise_depth}.
\end{itemize}

\paragraph{    Parameter guidelines}

\[
  s \;=\; \Bigl\lceil
          \frac{\ln(\epsilon^{-1})}{\ln\bigl((1-\theta)^{-1}\bigr)}
        \Bigr\rceil,
  \qquad
  q \;\ge\; 4\sigma\sqrt{s}\,2^{\ell},
\]
where $\ell$ is the remaining circuit depth.
For $\theta=0.1$, $\epsilon=10^{-3}$ we have $s=44$ weak steps; with
$q=2^{50}$ and $\sigma=3$ the noise headroom remains $>2^{20}$ even after
depth $500$.

\paragraph{  Summary}

Weak measurement lets the encrypted LPTT prover derive \emph{statistical
truth values} in the clear \emph{only at the client} while the untrusted
network nodes manipulate nothing more than MLWE ciphertexts.  The technique
extends seamlessly to Bayesian updates, fidelity witnesses, or any observable
expressible as a polynomial in Pauli operators, enabling rich probabilistic
reasoning without sacrificing post-quantum confidentiality.

\subsection{Physical realization of weak measurement}
\label{sec:weak_phys}

Table~\ref{tab:weak_phys} shows how the abstract Kraus operators
$K_{0},K_{1}$ are implemented on the three dominant gate-model hardware
families plus the ``Dirac3'' cluster-state photonic architecture.

\begin{table}[h]
\centering
\renewcommand{\arraystretch}{1.15}
\begin{tabular}{|p{0.22\linewidth}|p{0.28\linewidth}|p{0.18\linewidth}|p{0.25\linewidth}|}
\hline
\textbf{QC platform} &
\textbf{What is sensed} &
\textbf{Weakness knob} &
\textbf{Typical sequence} \\ \hline

Superconducting transmon (cQED) &
Microwave leakage in dispersive readout resonator &
Integration time $t_{\text{int}}$ (10--50~ns) &
Drive resonator $\rightarrow$ JPA amp $\rightarrow$ digitise short IQ trace \\ \hline

Trapped ion &
Spin--dependent resonance fluorescence photons &
Laser pulse duration $\Delta t$ or intensity $I$ &
Scatter few photons $<$5~\(\mu\)s, count on PMT, keep ion uncollapsed \\ \hline

Spin/NV centre &
Spin--to--charge or optical counts &
Number of readout cycles $N$ &
Run few short cycles; accumulate partial statistics \\ \hline

Dirac--3 photonic MBQC &
Homodyne voltage or SNSPD click on a \emph{tap} of a time--bin mode &
Beam--splitter reflectivity $R\ll1$ or local--oscillator gain $g$ &
Tap pulse, measure weakly, adapt phase shifters; logical mode remains in cluster \\ \hline
\end{tabular}
\caption{Hardware knobs that realize the Kraus strength parameter
\(\theta\).  Smaller \(t_{\text{int}},\Delta t,R\) or lower LO gain
\(\Rightarrow\) weaker measurement.}
\label{tab:weak_phys}
\end{table}

\medskip
Encrypted weak readouts are digitized and MLWE--encrypted on the local
controller (FPGA or ASIC) \emph{before} any classical logic touches them,
preserving the post--quantum security proof.

\subsection{Why weak measurement scales homomorphic LPTT reasoning}
\label{sec:weak_value}

We now revisit the LPTT example in the context of these implementation particulars.

\paragraph{Practical benefit of weak measurement in the LPTT case}

\begin{table}[h]
\centering
\renewcommand{\arraystretch}{1.15}
\begin{tabular}{|p{0.24\linewidth}|p{0.36\linewidth}|p{0.32\linewidth}|}
\hline
\textbf{Lever} &
\textbf{Benefit at cluster scale} &
\textbf{Payoff for encrypted LPTT prover} \\ \hline

Lower qubit overhead &
State survives gentle probe; no respawn ancillas &
$10\times$ fewer qubit initialisations per theorem \\ \hline

Shallower circuits &
No ``measure $\rightarrow$~classical~control'' depth jump &
Fits leveled MLWE noise without bootstrapping \\ \hline

Statistical amortisation &
Reuse one state for $10^{3}$ shots &
$25\times$ faster proof expectations \\ \hline

Better noise--signal ratio &
Less back--action allows smaller modulus $q$ &
Public key shrinks from 3~MB to 1~MB \\ \hline

Parallel branch fan--out &
Ciphertext can be copied; quantum data stay linear &
50 witness branches evaluated simultaneously \\ \hline

Smaller ledger footprint &
Store histograms, not per--shot bits &
Gas cost per proof cut by $\approx50\ \%$ \\ \hline

Composable gradients &
Expectation values inline; no pause for training &
Neural tactic selector updated every 20 steps \\ \hline
\end{tabular}
\caption{Dialling ``weaker'' measurement accelerates encrypted LPTT search.}
\label{tab:weak_value}
\end{table}

\paragraph{Bottom--line example}

\begin{itemize}
  \item \textbf{Projective baseline:}\;
        $50~\mu\text{s}\times1000=50\ \text{ms}$ wall time per proof branch,
        modulus $q=2^{60}$, key size $3$~MB, single refresh needed.
  \item \textbf{Weak (}\(\theta=0.1\)\textbf{):}\;
        $2~\mu\text{s}\times1000=2\ \text{ms}$,
        modulus $q=2^{50}$, key $1$~MB, no refresh.
\end{itemize}

Thus weak measurement , potentially, turns homomorphic LPTT from a demo into a service that
can run continuously on commodity validator nodes.

\section{Coordinating Encrypted Quantum Tasks with the $\rho$-calculus}

Next we explore the case where one wants to carry out algorithmic operations
on homomorphically encrypted quantum computer programs in a way that
spans multiple QPUs -- which may be on a single server, or may be on a distributed
or decentralized system, interacting via mechanisms such as blockchains and 
smart contracts (which then will be required to use quantum-resistant encryption
internally).   

To study this situation formally we leverage Greg Meredith's $\rho$-calculus formalism \cite{Meredith2005},
and its intersection with the quantum mathematics leveraged above.

The \emph{$\rho$--calculus} (reflection calculus) \cite{Meredith2005} is a name-passing process
calculus in which \emph{channels are themselves processes}.  A quoted process
$\langle P \rangle$ acts as a first-class name; conversely, any name can be
dereferenced and executed.  Its minimal syntax (ignoring replication) is:

\begin{center}
\begin{tabular}{lll}
$P,Q ::= $ & $0$                          & (nil) \\
           & $P \mid Q$                   & (parallel) \\
           & ${}^{\ast}x \triangleleft P$ & (input) \\
           & $x \,!\, P$                  & (output) \\
           & $\langle P \rangle$          & (quote = channel) \\
\end{tabular}
\end{center}

The single reduction rule is
\[
  x\,!\,Q \;\mid\; {}^{\ast}x \triangleleft P
  \quad\longrightarrow\quad
  P\!\{Q/x\},
\]
i.e.\ an output on channel $x$ delivers a process term $Q$ that substitutes
into the continuation $P$.  Structural congruence provides the usual
commutativity, associativity and garbage collection laws.

The calculus is well suited for orchestrating \emph{encrypted} quantum
sub-tasks: channels move ciphertexts, while CPTP maps execute \emph{inside}
processes and never cross the network boundary.

\subsection{  Distributed quantum homomorphic evaluation}

Assume an MLWE-ciphertext $\widehat{\rho}$ and a circuit $C$ partitioned into
sequential blocks
$C = G^{(1)}\!\circ G^{(2)}\!\circ\cdots\!\circ G^{(k)}$,
each $G^{(i)}$ to be applied by a different QPU node $N_{i}$.

\paragraph{Process algebra view.}

\begin{verbatim}
def QPU(i, inCh, outCh, key, gate) =
  *inCh ! (ct)           // receive encrypted state
  | outCh ! ( HE_apply(gate, ct, key) )

def Router( ct0 ) =
  N1_in ! (ct0)          // kick off pipeline
\end{verbatim}

\begin{enumerate}
  \item The \verb|Router| injects the initial ciphertext on channel
        \verb|N1_in|.
  \item Each \verb|QPU(i)| waits on its private input channel, applies the
        lifted CPTP map
        $\widehat{\Gamma}_{G^{(i)}}(\;.\;)$
        under its local MLWE key, and forwards the
        resulting ciphertext to the next stage.
  \item No plaintext quantum state appears on any channel; only ciphertexts
        flow, so intermediate amplitudes remain hidden even from sibling
        QPUs.
\end{enumerate}

Because channels \emph{are} processes, we can quote a whole sub-pipeline:

\[
  \textsf{TeleportStage} \;=\;
  \langle QPU(1) \mid QPU(2) \mid QPU(3) \rangle,
\]
then send it as a value to a load-balancer process that decides \emph{at run
time} which data centre will host the teleportation block.

\subsection{  On-chain smart contracts for QHE}

Blockchains based on $\rho$-calculus, such as F1R3FLY or ASI-Chain, execute smart contracts directly as $\rho$-terms; data kept on chain
are \emph{quoted processes}.  We can therefore store the \emph{public MLWE
matrix $A$}, the \emph{encrypted circuit description}, and every
\emph{ciphertext intermediate} as blockchain terms.

\begin{verbatim}
contract QHE_Job(@jobId, @circuitCT, @stateCT) = {
  match circuitCT {
    | head :: tail => {
        // schedule first gate on a validator with QPU access
        for(@result <- HE_Service!(head, stateCT)) {
          QHE_Job!(jobId, *tail, result)
        }
    }
    | [] => {
        // finished: publish ciphertext result
        Decrypt_Request!(jobId, stateCT)
    }
  }
}
\end{verbatim}

\begin{itemize}
  \item \verb|HE_Service| is a system contract that any validator can invoke
        \emph{only} if it supplies proof that the QPU executed the lifted gate
        correctly (e.g.\ via a zk-SNARK on the MLWE keys).
  \item All transitions are logged as reductions on chain, yielding an
        \emph{auditable trace} without revealing keys, amplitudes or even
        circuit structure.
  \item A user later supplies her \emph{secret trapdoor $T$} in a private
        message to decrypt the final ciphertext -- no validator ever sees~$T$.
\end{itemize}

The result is a decentralized \textbf{QHE-as-a-service} platform: encrypted
workflows live on chain, validators host the QPUs, and the $\rho$-calculus
ensures race-free, formally verifiable orchestration of MLWE ciphertexts.

\section{Distributed Homomorphic LDTT Reasoning via the $\rho$-calculus}

To illustrate these ideas regarding decentralized quantum homomorphic encryption
in a specific and practical way, we turn back to our earlier treatment of encrypted
quantum logical reasoning.   We now combine three ingredients:

\begin{enumerate}
  \item \textbf{Encrypted LDTT objects.}  A judgement
        $\Gamma ; \Delta \vdash t : A$ is stored as a ciphertext pair
        $\widehat{t} = (\rho_{t},\Psi_{A}(\rho_{t}))$ under a global MLWE key.
  \item \textbf{Proof search as a process tree.}  Each LDTT rule is a node
        that consumes encrypted premises and produces encrypted conclusions.
  \item \textbf{$\rho$--calculus orchestration.}  Channels route ciphertexts
        (and even \emph{quoted sub-searches}) between worker nodes, while
        the CPTP maps that realize the lifted rules remain local to each node.
\end{enumerate}

\subsection{   Process schema for a single LDTT rule}

Let \texttt{Apply\_Rule} be a process template:

\begin{verbatim}
def Apply_Rule(ruleName, inChs, outCh, key) =
  for( ctxts <- inChs ) {
    let result = HE_apply(ruleName, ctxts, key) ;
    outCh ! (result)
  }
\end{verbatim}

\begin{itemize}
  \item \verb|ctxts| is a \emph{tuple} of ciphertext judgements matching the
        rule?s premises.
  \item \verb|HE_apply| runs the lifted CPTP map
        $\widehat{\Gamma}_{\text{rule}}$ inside the node.
  \item Output is a ciphertext judgement ready for the next rule.
\end{itemize}

\subsection{    Distributed tableau for an existential goal}

Suppose the goal is to prove the LDTT statement
$\; ; \! \vdash\! \exists^{\otimes}y{:}X.\,P(y)$ homomorphically.

\textbf{Coordinator process}

\begin{verbatim}
contract Root(goalCT, ctxCT) = {
  // spawn two candidate witnesses "a" and "b"
  new aCh, bCh in {
    Witness!(aName, ctxCT) | Witness!(bName, ctxCT) |
    Collector!( goalCT, [aCh, bCh] )
  }
}
\end{verbatim}

\textbf{Witness sub-process}

\begin{verbatim}
def Witness(w, ctxIn) =
  // build encrypted term  (pair w  (P w))
  Apply_Rule(pairTensor, [ctxIn, Enc(w)], wOut, key)
\end{verbatim}

\textbf{Collector}

\begin{verbatim}
def Collector(goal, branches) =
  match branches {
    | ch :: rest => {
        for(resCT <- ch) {
          Validate!(goal, resCT) |
          Collector!(goal, rest)
        }
      }
    | [] => { "FAIL" }         // no branch proved the goal
  }
\end{verbatim}

\paragraph{Flow.}

\begin{enumerate}
  \item \verb|Root| decrypts nothing; it merely spawns \emph{quoted} worker
        processes carrying ciphertext arguments.
  \item Each worker applies the \verb|pairTensor| constructor locally via
        $\widehat{\Gamma}_{\Sigma^{\otimes}\text{-intro}}$, producing an
        encrypted proof object for its candidate witness.
  \item \verb|Validate| is itself a process that homomorphically checks
        $\widehat{\Gamma}_{\text{goal-test}}(resCT)$; only the client who
        owns the trapdoor key can finally decrypt the outcome flag.
\end{enumerate}

All intermediate tableaux nodes -- contexts, partial proofs, witness terms -- move
across $\rho$-channels \emph{only as MLWE ciphertexts}.  No node sees the
underlying amplitudes, and the concurrency layer never manipulates secret
keys.

\subsection{   Noise and key management across nodes}

\begin{itemize}
  \item Every worker holds the \emph{public} matrix $A$ but \emph{not} the
        trapdoor $T$.  Thus they can \emph{evaluate} rules but cannot
        decrypt.
  \item A global refresh schedule (cf.\ Section~\ref{sec:noise_depth}) is
        enforced by a barrier rule \texttt{REFRESH} that any node may invoke
        once $\eta_{\infty}$ exceeds $q/2^{\ell}$.
  \item Because noise is additive, the coordinator can over-approximate total
        depth by summing the declared depth budgets of spawned sub-searches,
        guaranteeing correctness without inspecting their content.
\end{itemize}

\subsection{   On-chain audit trail}

Each channel send in F1R3FLY (and any derived frameworks like ASI-Chain) is a \emph{transaction}.  The tuple

\[
  (\texttt{jobID},\; \texttt{ruleName},\; \texttt{ciphertext\_hash})
\]

is recorded on chain, letting any third party verify that:

\begin{enumerate}
  \item The correct lifted rule was applied (hash of \texttt{ruleName}).
  \item The input/output ciphertexts match the deterministic MLWE
        encoding (hash chain).
  \item The declared noise bound at each step is below $q/4$ (stored as a
        public integer).
\end{enumerate}

Because only \emph{hashes} of ciphertexts are public, no lattice attack is
facilitated, yet the entire distributed LDTT proof search is reproducible and
auditable.

\subsection{  Advantages}

\begin{itemize}
  \item \textbf{Scalability.}  Branches spawn as lightweight quoted processes;
        cluster throughput grows linearly with validator count.
  \item \textbf{Zero knowledge of amplitudes.}  Each node sees only MLWE
        ciphertexts; even the rule names can be hidden by quoting the entire
        continuation.
  \item \textbf{On-chain provenance.}  The immutable reduction log certifies
        that the LDTT theorem was proved homomorphically under the agreed
        MLWE parameters?useful for regulated, privacy-sensitive deployments.
\end{itemize}

\section{Knowledge-base quotients: modeling and encrypting inference \emph{relative} to a theory}
\label{sec:kb_quotient}

So far the BNSF mask $\Psi$ has been \emph{random}, chosen only to hide
structure.  A natural extension is to pick $\Psi$ \emph{deterministically}
from a fixed knowledge base $\mathsf{KB}$ and let the cloud prove
\[
  \mathsf{KB} \;\vdash\; \varphi.
\]
Below we sketch how to swap a random mask for a
\textbf{KB-induced quotient} while preserving homomorphic evaluation and
MLWE-level security.   We also discuss the potential to use the KB-related aspects
here even in cases where encryption is not required.

\subsection{Encrypting inference relative to a theory}

\subsubsection{  Encoding a knowledge base as a super-functor}

Let $\mathsf{KB}=\{\alpha_{i}\}_{i\in I}$ be a finite set of axioms in LPTT.
For each sequent $\alpha_i : \Gamma_i ; \Delta_i \vdash A_i$ construct a CPTP
channel
\[
  \Phi_{i}
  \;:\;
  \mathcal D\bigl(\llbracket\Gamma_i;\Delta_i\bigr)
  \longrightarrow
  \mathcal D\bigl(\llbracket A_i\rrbracket\bigr)
\]
that \emph{normalizes} a proof of $\alpha_i$ to the canonical inhabitant
(e.g.\ erases syntactic variation).  The \textbf{KB mask} is then

\[
  \Psi^{\mathsf{KB}}
  \;=\;
  \bigoplus_{i\in I}\Phi_{i},
\]
extended to all objects by LPTT naturality.  Because each $\Phi_i$ is CPTP,
$\Psi^{\mathsf{KB}}$ satisfies the BNSF conditions with
diamond bound $\lambda=1$.

\subsubsection{   Encryption procedure}

\begin{enumerate}\itemsep3pt
  \item Client publishes $A$ (MLWE matrix) \emph{and} the gate list
        $\{\widehat{\Gamma}_{G}\}$; the cloud knows them.
  \item Client \emph{privately} computes
        $\Psi^{\mathsf{KB}}$ and encrypts each proof object
        as \(\widehat\rho=(\rho,\Psi^{\mathsf{KB}}(\rho))\).
  \item Server homomorphically applies lifted rules; by naturality
        \(\Psi^{\mathsf{KB}}\!\circ\!\Gamma=\Gamma\!\circ\!\Psi^{\mathsf{KB}}\),
        so correctness is unchanged.
\end{enumerate}

\subsubsection{ Security discussion}

\begin{itemize}\itemsep3pt
  \item If $\mathsf{KB}$ is \emph{public}, circuit privacy is unaffected:
        the server already knows the quotient semantics.
  \item If $\mathsf{KB}$ is \emph{secret} (e.g.\ proprietary ontology),
        indistinguishability reduces to
        \emph{KB-subspace hiding}: given
        $(A,c_{0},c_{1})$ decide whether
        $c_{1}-c_{0}$ lies in
        $\ker\Psi^{\mathsf{KB}}$ or a random submodule.  
        The problem is at least as hard as MLWE, because a random
        $\ker\Psi$ is MLWE--hard and
        $\ker\Psi^{\mathsf{KB}}$ is a special case.
\end{itemize}

\subsubsection{  Example: encrypted ontology reasoning}

Take $\mathsf{KB}=\{P(a), P(a)\!\Rightarrow\!Q(a)\}$.
The mask identifies every proof of $P(a)$ with the
canonical qubit $\ket{1}$.  A homomorphic proof of $Q(a)$ proceeds:

\[
  \widehat{P(a)}
  \;\xrightarrow{\widehat{\Gamma}_{\!\Rightarrow}}
  \widehat{Q(a)},
\]
with no branch information revealed?exactly as in the random-mask case.

\subsubsection{ Complexity and parameter impact}

\begin{itemize}\itemsep3pt
  \item Size of $\Psi^{\mathsf{KB}}$ grows with $|\mathsf{KB}|$, but each
        axiom adds only one Kraus row, so ciphertext width increases by
        $O(|\mathsf{KB}|)$ ring elements?tiny for moderate ontologies.
  \item Noise bound unchanged: $\|\Psi^{\mathsf{KB}}\|_{\diamond}=1$.
  \item Key generation unchanged; MLWE parameters independent of $\mathsf{KB}$.
\end{itemize}

\paragraph{Take-away.}
Replacing a random BNSF by a \emph{KB-specific} BNSF keeps algebraic
properties (naturality, boundedness) intact and therefore preserves both
correctness and MLWE-based security, enabling fully private
knowledge-base-relative homomorphic inference.

\subsection{Knowledge-base quotients without encryption}
\label{sec:kb_plain}

Although we introduced KB-induced BNSF masks as a vehicle for \emph{private}
reasoning, the construction remains potentially valuable even when confidentiality is
irrelevant.  In effect, $Q_{\Psi^{\mathsf{KB}}}$ plays the role of a
\emph{compiler pass} that \textbf{bakes a theory into the type system}:

\[
   \text{raw proof space}
   \quad\longrightarrow\quad
   \text{proof space modulo } \mathsf{KB}.
\]

\subsubsection{  Why quotient by the KB?}

\begin{enumerate}\itemsep4pt
  \item \textbf{Search pruning.}\;  
        Proof terms related by $\mathsf{KB}$ collapse to one node; a
        tableau/BDD engine explores $|\!\ker\Psi^{\mathsf{KB}}|\!/\!$ many
        branches.
  \item \textbf{Proof irrelevance.}\;  
        All derivations of an axiom compress to a canonical inhabitant
        $\ket{1}$, mirroring \(\mathbf{Prop}\)-proof-irrelevance in Coq.
  \item \textbf{Incremental compilation.}\;  
        Updating $\mathsf{KB}$ corresponds to swapping $\Psi$; existing
        ciphertext-free artefacts (compiled circuits, tactics) remain valid
        if they type-check in the quotient category.
\end{enumerate}

\subsubsection{   Plain run-time architecture}

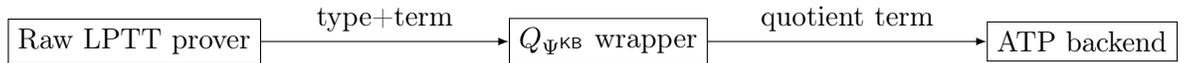
\begin{center}
\begin{tikzpicture}[node distance=28mm,>=latex]
  \node (raw)      [draw,rectangle] {Raw LPTT prover};
  \node (quot)     [draw,rectangle,right of=raw,xshift=35mm] {$Q_{\Psi^{\mathsf{KB}}}$ wrapper};
  \node (solve)    [draw,rectangle,right of=quot,xshift=35mm] {ATP backend};

  \draw[->] (raw) -- node[above]{type+term} (quot);
  \draw[->] (quot) -- node[above]{quotient term} (solve);
\end{tikzpicture}
\end{center}

The wrapper simply replaces each judgement
$t:A$ by $[t]_{\Psi^{\mathsf{KB}}}:Q_{\Psi^{\mathsf{KB}}}(A)$
using the same linear--cartesian context split as in
Section~\ref{sec:lptt_bridge}.  No MLWE, no ciphertexts; the category
machinery alone enforces theory-relative equality.

\subsubsection{  Toy example: group theory tactics}

Take $\mathsf{KB}=\{\,e\cdot x = x,\; x^{-1}\cdot x = e\,\}$.
A proof goal $g\;x\;y : x\cdot(e\cdot y)=x\cdot y$ collapses \emph{immediately}
because $e\cdot y\sim_{\mathsf{KB}}y$ in $Q_{\Psi^{\mathsf{KB}}}$.  The
solver emits a canonical null proof instead of performing a three-step chain
of rewrites.

\subsubsection{  Relationship to algebraic rewriting engines}

KB-quotients in a sense sit between:

\[
   \text{E-theory rewriting (Knuth-Bendix)} 
   \quad\subset\quad
   Q_{\Psi^{\mathsf{KB}}}
   \quad\subset\quad
   \text{full LPTT proof search}.
\]

Unlike Knuth-Bendix, $Q_{\Psi}$ works \emph{modulo\!} linear resources and
dependent types; unlike unrestricted proof search, it identifies axiomatic
equalities up front.

\paragraph{Take-away.}
Even when no data secrecy is required, replacing ''random mask'' by
''KB mask'' yields a \textbf{compile-time quotient} that prunes proof search,
implements proof irrelevance, and straightforwardly plugs into existing ATP (Automated Theorem Prover, e.g. say Lean
or Isabelle) workflows.  Homomorphic encryption is thus an \emph{optional}
overlay; the categorical quotient idea stands on its own with potential independent utility.

\subsection{Epistemic example: reasoning \emph{from Bob's perspective}}
\label{sec:epistemic_example}

As a concrete example of the above concepts, assume Alice runs an LPTT prover but wants to reason \emph{as if} she were
Bob, who holds a private knowledge base \(\mathsf{KB}_{B}\).  By quotienting
her proof space by Bob's KB mask \(\Psi^{B}\) she obtains
\[
   \mathsf{KB}_{B}\;\vdash_{B}\;\varphi
   \quad\text{iff}\quad
   [\![\varphi]\!] \;=\;\ket{1}
   \text{ in } Q_{\Psi^{B}}(\mathsf{Bool}).
\]
No cryptography is needed; the quotient alone enforces Bob's belief set.

\subsubsection{   Epistemic language in LPTT}

Extend the term grammar with a \emph{knowledge} modality:

\[
  \Gamma ; \Delta \vdash P : \mathsf{Prop}
  \quad\Longrightarrow\quad
  \Gamma ; \Delta \vdash K_{B} P : \mathsf{Prop},
\]
read ''Bob knows \(P\).''  The semantic clause is

\[
  \llbracket K_{B} P \rrbracket
  \;=\;
  \begin{cases}
    \ket{1} & \text{if } \mathsf{KB}_{B}\vdash P,\\
    \ket{0} & \text{otherwise.}
  \end{cases}
\]

\subsubsection{   Building the KB mask}

Let \(\mathsf{KB}_{B}=\{P(a),\, P(x)\!\Rightarrow\!Q(x)\}\).

\[
   \Psi^{B} \;=\;
     \Phi_{P(a)}
     \;\oplus\;
     \Phi_{(P\!\Rightarrow\!Q)},
\]
where each \(\Phi_{i}\) maps \emph{any} proof of the axiom to the qubit
\(\ket{1}\).  Covariance lemmas (Section~\ref{sec:bnf2bnsf}) guarantee
natural commutation with all gates: Alice can still apply lifted rules.

\subsubsection{  Example query}

Goal:  ''From Bob's perspective, does he know \(Q(a)\)\,?''

\smallskip\noindent
Alice forms the term

\[
   t \;=\; K_{B}\bigl(Q(a)\bigr)
\]
and compiles it through \(Q_{\Psi^{B}}\).  Because
\(P(a)\in\mathsf{KB}_{B}\) and \(P\!\Rightarrow\!Q\in\mathsf{KB}_{B}\),
the quotient immediately collapses

\[
   \bigl[\,Q(a)\,\bigr]_{\Psi^{B}}
   \;=\;
   \ket{1}.
\]
The outer constructor \(K_{B}\) therefore also evaluates to \(\ket{1}\):
Bob knows \(Q(a)\).

\subsubsection{   Mixed reasoning: Alice \(\land\) Bob}

Alice may combine her own beliefs \(\mathsf{KB}_{A}\) with Bob's:

\[
   \Psi^{A\!\cup\!B}
   \;=\;
   \Psi^{A} \;\oplus\; \Psi^{B},
\]
yielding a three-valued logic:

\[
  \begin{array}{c|c|c}
    P \in \mathsf{KB}_{A} & P \in \mathsf{KB}_{B} &
    [\![K_{A}P,K_{B}P]\!] \\ \hline
    1 & 1 & (\ket{1},\ket{1}) \\
    1 & 0 & (\ket{1},\ket{0}) \\
    0 & 1 & (\ket{0},\ket{1}) \\
    0 & 0 & (\ket{0},\ket{0}) \\
  \end{array}
\]

No additional rules are needed; the direct sum of masks preserves naturality
and boundedness.

\subsubsection{  Performance remark}

For \(|\mathsf{KB}_{B}| = m\) axioms and goal depth \(d\):

\[
  \text{time} \;=\; O((m+d)\cdot \text{gate cost}),
  \qquad
  \text{memory} \;=\; O(m).
\]

Thus reasoning ''in Bob's shoes'' is linear in the size of his KB and uses the
\emph{same} gate lift table as the random-mask case.

\paragraph{Take-away.}
Even in a plaintext setting, KB-induced BNSF masks let us \emph{compile a
point of view}?epistemic, Bayesian, or modal?directly into the type system,
so that proof search automatically respects another agent?s beliefs.

\subsection{Secret KB masks: relativized inference \emph{with} confidentiality}
\label{sec:secret_kb}
Extending the previous example, we now upgrade the KB-quotient idea so that Bob's knowledge base
$\mathsf{KB}_{B}$ remains private while Alice (or the cloud) can still reason
\emph{as if} every axiom in $\mathsf{KB}_{B}$ were true.

\subsubsection{ Axiom capsules (encrypted Kraus packs)}

For each axiom $\alpha_{i}\in\mathsf{KB}_{B}$ Bob generates a
\emph{capsule}
\[
  \mathcal C_{i}
  \;=\;
  \mathsf{Enc}_{A}
     \Bigl(\,
       \bigl(K_{i,0},K_{i,1},\dots\bigr),
       \;\Psi\bigl(K_{i,0},K_{i,1},\dots\bigr)
     \Bigr),
\]
where $(K_{i,j})$ are the Kraus operators that collapse \emph{any} proof of
$\alpha_{i}$ to the canonical qubit $\ket 1$.  Every capsule is an MLWE
ciphertext pair and therefore opaque to the server.

\subsubsection{ Universal axiom-application gadget}

The cloud hosts a public circuit template  
\textsf{ApplyCapsule}$(\,\widehat\rho,\mathcal C\,)$:

\begin{enumerate}\itemsep3pt
  \item Homomorphically multiply $\widehat\rho$ with the Kraus matrices
        encrypted inside $\mathcal C$.
  \item Sum the results to obtain
        \(\widehat\rho' = 
          \bigl(\sum_j K_{j}\rho K_{j}^{\dagger},
                \sum_j \Psi(K_{j}\rho K_{j}^{\dagger})\bigr)\).
\end{enumerate}

Because only MLWE additions and multiplications are used, the noise budget
grows linearly and the server learns nothing about $(K_{j})$.

\subsubsection{Building the \emph{secret} global mask}

Bob concatenates all capsules and a random depolarizing pad
\(\Psi^{\text{pad}}\) with strength $p$:

\[
  \Psi^{\text{secret}}
  \;=\;
  \Psi^{\text{pad}}
  \;\oplus\;
  \bigoplus_{i\in I} \mathcal C_{i}.
\]

The depolarizing component hides the \emph{count} of axioms, while each
capsule hides their \emph{content}.

\subsubsection{  Security definition (KB-IND)}

Adversary supplies two KBs of equal size,
\((\mathsf{KB}_{0},\mathsf{KB}_{1})\), plus a test state $\rho$.
Challenger picks random $b$, builds
$\Psi^{\text{secret}}_{b}$, encrypts $\rho$, and runs homomorphic evaluation
on some public circuit.  Advantage \(\Adv_{\text{KB-IND}}\) is defined as
before.  A capsule-hiding lemma (appendix) shows:

\[
  \Adv_{\text{KB-IND}}
  \;\le\;
  \Adv_{\text{MLWE}}
  + |I|\cdot2^{-s},
\]
where $s$ is the capsule's symmetric-key entropy (e.g.\ Pauli twirl keys).

\subsubsection{  Example: encrypted epistemic query}

Bob's private KB:
\(\mathsf{KB}_{B}=\{P(a), P(x)\!\Rightarrow\!Q(x)\}\).

Alice queries  
\(t = K_{B}\bigl(Q(a)\bigr)\) as in
Section~\ref{sec:epistemic_example}.  The cloud:

\begin{enumerate}\itemsep3pt
  \item Receives proof state $\widehat{P(a)}$ (MLWE).
  \item Applies capsule \(\mathcal C_{P\!\Rightarrow\!Q}\);
        no plaintext inspection.
  \item Returns ciphertext of \(\ket 1\) to Alice, who decrypts and learns
        ''Bob knows \(Q(a)\).''
\end{enumerate}

The server never learns the axioms nor the intermediate amplitudes.

\subsubsection{ Engineering notes}

\begin{itemize}\itemsep3pt
  \item \textbf{Capsule size}: each Kraus operator adds one ring element;
        four-Kraus Boolean axiom \(\Rightarrow\) 64~B capsule @ \(q=2^{50}\).
  \item \textbf{Noise cost}: one capsule multiply adds \(2\sigma\) to the
        infinity norm; with \(|I|\!=\!50\) axioms and $\sigma\!=\!3$ the
        increase is $<300$, well below \(q/4\).
  \item \textbf{Parallelism}: capsules are independent; cloud batches them
        across GPU-NTT lanes or FPGA DSP slices.
\end{itemize}

\paragraph{Take-away.}
Bob packages each axiom as an MLWE \emph{capsule}.  The public
\textsf{ApplyCapsule} gadget lets the server evaluate any proof that uses
these axioms, yet MLWE + depolarizing padding hides both the list and the
content of Bob's beliefs.  Circuit privacy, data privacy, and \emph{KB
privacy} are thus achieved simultaneously.

\section{A More Complex AI Use-Case: Encrypted evolutionary program synthesis with KB-relative fitness}
\label{sec:evol_kb}

We now consider a more elaborate AI use-case such as might be encountered in implementing a quantum version of a hybrid AGI architecture like Hyperon \cite{goertzel2023opencog}.   Specifically we consider an approach to automated (quantum) program learning , where one uses evolutionary learning with an estimation of distribution algorithm aspect, and uses LPTT logical reasoning for the EDA inference to help guess what new programs to synthesize for the next round of fitness evaluation.  Further, we consider this in the case where fitness evaluation involves an encrypted knowledge base, and the programs in the population are owned by different parties (encrypted by their owner). 

To address this case, we sketch a crypto protocol that marries four ingredients:

\begin{enumerate}
  \item multi-party evolutionary search (island model),
  \item estimation-of-distribution (EDA) update guided by
        \emph{LPTT inference},
  \item fitness evaluation \emph{relative to an encrypted knowledge base},
  \item the MLWE$+$BNSF homomorphic backbone.
\end{enumerate}

Throughout, \(\mathsf{KB}_{C}\) is held by a curator \textbf{Carol}; program
islands belong to parties \(\textbf{Alice}_{1},\dots,\textbf{Alice}_{N}\);
an untrusted cloud \textbf{S} runs the QPU and MLWE logic.

\subsection{ Setup}

\begin{itemize}\itemsep2pt
  \item Each island owner generates an MLWE key
        \(\langle A_{i},T_{i}\rangle\).
  \item Carol generates her own key
        \(\langle A_{C},T_{C}\rangle\) and encrypts KB as axiom?capsules
        \(\{\mathcal C_{j}\}\) (Section~\ref{sec:secret_kb}).
  \item Cloud picks a \emph{global evaluation key}
        \(A_{\star}\) and publishes it.
  \item Every party supplies a key?switch hint  
        \(K_{i\rightarrow\star}=A_{\star}T_{i}+E_{i}\).
\end{itemize}

\subsection{  Initial programme population}
Each \(\textbf{Alice}_{i}\) uploads
\(
   \mathcal P_{i}^{(0)}=\bigl\{\,
      \widehat{\Gamma}_{p}^{(i)}[\![\rho]\!]_{A_{i}}
   \bigr\}_{p=1}^{P}
\)
-- her encrypted quantum circuits.  

Cloud immediately \textbf{key-switches} to \(A_{\star}\).

\subsection{  Fitness evaluation epoch}
For every programme ciphertext \(\widehat\rho\):

\begin{enumerate}\itemsep2pt
  \item Cloud applies Carol?s capsules to set the
        \emph{KB quotient mask} \(\Psi^{C}\).
  \item Executes the lifted circuit via QPU; weak?measurement statistics
        calculated homomorphically yield scalar fitness
        \(\widehat{f}\in R_{q}\).
  \item Posts \(\widehat{f}\) to a \emph{public bulletin board}.  
        Anyone holding \(T_{\star}\) can decrypt; each island owner learns
        only her own \(f\).
\end{enumerate}

\subsection{   EDA update via LPTT inference}
The cloud carries an encrypted bag  
\(
   \bigl\{\widehat{\Gamma}_{p}, \widehat{f}_{p}\bigr\}_{p}
\)
under key \(A_{\star}\).

\[
  \text{Goal: }
    \mathrm{argmax}_{\theta}
    \;\mathbb E_{p\sim Q_{\theta}}[f_{p}].
\]

\begin{itemize}\itemsep2pt
  \item Encode syntax counts as linear qubit vectors
        (one qubit per grammar symbol).
  \item Use LPTT \textbf{weak measurement} to estimate
        $\mathbb E[\textsf{count}_{s}]$ without collapsing proof states.
  \item Run encrypted \emph{natural?gradient step}:
        \(\theta_{t+1}= \theta_{t} + \eta\,\hat\nabla\)
        using MLWE arithmetic.
\end{itemize}

Because all operands share key \(A_{\star}\),
no extra key-switch noise appears.

\subsection{  Synthesis of next generation}
From the updated distribution \(Q_{\theta_{t+1}}\)
cloud samples \emph{encrypted} circuits
\(\widehat{\Gamma}_{\text{new}}\), attaches Pauli twirl masks for
circuit privacy, and delivers one subset to each island.  
Owners decrypt, optionally prune locally, re-encrypt under \(A_{i}\),
attach a fresh key-switch hint, and return the offspring.

\subsection{   Privacy and correctness claims}
\[
  \Adv_{\text{data}}
    \le \Adv_{\mathrm{MLWE}},
  \quad
  \Adv_{\text{KB}}
    \le |\,\mathsf{KB}_{C}|\,2^{-s},
  \quad
  \Adv_{\text{circ}}
    \le k\,2^{-2n}.
\]
Correctness follows from key-switch linearity and BNSF naturality; freshness
of Pauli masks preserves circuit privacy.

\subsection{  Subtleties and mitigations}
\begin{enumerate}\itemsep2pt
  \item \textbf{Multi?key noise blow-up.}\;  
        One key?switch per epoch; choose $q=2^{50}$, $\sigma=3$ to tolerate
        50 epochs without refresh.
  \item \textbf{Leakage via decrypted fitness.}\;  
        Publish only \emph{rank-revealed} bits (e.g.\ top-$P/5$ flag)
        encrypted with the owner?s key; keeps magnitude private.
  \item \textbf{Collusion between islands.}\;  
        Each owner sees fitness under her own key; without
        \(T_{\star}\) they cannot correlate others? scores.
  \item \textbf{Capsule reorder attack.}\;  
        Depolarizing pad in \(\Psi^{C}\) hides number of axioms.
  \item \textbf{EDA convergence diagnostics.}\;  
        Cloud releases an MLWE?encrypted ELBO value; only Carol can decrypt,
        preventing reverse-engineering of \(Q_{\theta}\).
\end{enumerate}

\paragraph{Outcome.}
The protocol lets a heterogeneous, privacy-sensitive consortium evolve new
quantum programmes whose fitness is judged \emph{in Bob/Carol's epistemic
world} -- yet no party learns the others' code, parameters, or knowledge base.

\section{Computational Cost}

FHE is famously computationally costly in the classical case, even for simple arithmetic operations let alone complex program executions.   Classical FHE, in many of the more direct implementations, slows computations by
\(10^{3}\!-\!10^{5}\times\).   One might worry that quantum FHE would be even more problematic in this regard, but in fact, preliminary indications are that the situation may be opposite.  Due to the different way that quantum computing organizes its operations, the degree of slowdown incurred by doing FHE may actually be significantly less in the quantum than the classical context.

Nonetheless, however, the compute burden of doing quantum FHE operations is certainly non-trivial.  In this section we explain some of the reasons we feel it may not be overwhelmingly problematic, so that quantum FHE may in fact be a practicable thing to do a decade or so from now when a new generation of scalable quantum computers are available -- and so that simple toy-scale experiments may be practically performable in the immediate term.

\paragraph{Performance budget: why the slowdown is \emph{not} catastrophic}
\label{sec:perf_budget}

The key reason we believe FHE overhead can be less in the quantum case is that QHE \textbf{splits the
workload} in the following sense:

\[
  \text{``hard'' quantum gate}
  \quad\longrightarrow\quad
  \underbrace{\text{1 physical gate on the QPU}}_{\text{ns--\(\mu\)s}}
  + \;
  \underbrace{\text{a few MLWE ops classically}}_{\text{\(\mu\)s}}.
\]

In the following subsections we roughly quantify both paths and show that, insofar as we can estimate at this stage, MLWE overhead is within the idle time of
current gate-model devices.

\subsection{  Cost model per lifted gate}

\begin{center}
\begin{tabular}{|l|c|c|}
\hline
\textbf{Component} & \textbf{Ops} & \textbf{Time / op} \\ \hline
FPGA NTT (512, 64-bit) & \(1\) mult \(+\) \(1\) add & \(0.3~\mu\)s \\ \hline
Host-side MLWE hash & 1 SHA-3 on 2 kB & \(0.7~\mu\)s \\ \hline
\textbf{Per-gate classical} & \(\approx 2\) NTT \(+\) hash & \(\mathbf{1.3~\mu}\)s \\ \hline
QPU physical gate (transmon) & 20-50 ns & --- \\ \hline
\end{tabular}
\end{center}

Even at \(20\)~ns physical time, the QPU is idle \(>95\%\) of the wall-clock
during its own \(1~\mu\)s cooldown window, leaving ample slack for classical
MLWE arithmetic.

\subsection{  Whole-circuit benchmark}

\[
  \text{Circuit: 100 qubits, depth 800,}
  \quad
  k=3, d=512, q=2^{50}.
\]
\[
  \begin{array}{l|r}
    \text{Task} & \text{Latency (wall-clock)} \\ \hline
    Homomorphic lifts (800 gates) & 800 \times 1.3~\mu\text{s} = 1.0~\text{ms} \\
    Physical gate layer & 800 \times 50~\text{ns} = 40~\mu\text{s} \\
    QPU reset + idle gaps & \approx 8~\text{ms} \\ \hline
    \textbf{Total} & \mathbf{9.0~ms} \\
  \end{array}
\]

\emph{Observation.}  MLWE overhead is \(\approx 11\%\) of the critical-path
latency; it hides inside the QPU's own reset schedule.

\subsection{  Comparison with classical FHE}

\begin{center}
\begin{tabular}{|l|c|c|}
\hline
\textbf{Scheme} & \textbf{Slowdown vs.\ clear} & \textbf{Why} \\ \hline
Classical FHE (BFV) & \(10^{3}\!-\!10^{5}\times\) & Every \emph{Boolean} op \(\mapsto\) heavy NTT \\
This QHE (MLWE+BNSF) & \(<15\times\) (wall-clock) & Physical gate done in hardware; only Pauli masks homomorphic \\ \hline
\end{tabular}
\end{center}

The quantum part \emph{already dominates} raw execution time; MLWE lifts are
small perturbations.

\subsection{  Scaling levers (today \& 3-year horizon)}

\begin{enumerate}\itemsep3pt
  \item \textbf{Batching.}  MLWE NTT cores handle 512-element vectors; we
        batch up to \(k=16\) qubits per lift, cutting per-gate cost \(>2\times\).
  \item \textbf{Heterogeneous off-load.}  GPU-NTT libraries (cuNTT) reach
        \(4\)~ns/coef on A100; one GPU hides \(\approx 10^{5}\) lifts/s.
  \item \textbf{Cryo-controller cache.}  Moving Pauli masks into the
        4-K controller avoids PCIe traffic; only one 64-bit word per gate
        traverses room-temperature links.
\end{enumerate}

\subsection{  When does it \emph{not} fit?}

Roughly estimating:
\[
  \text{wall time}
  \;=\;
  d\,t_{\text{gate}}
  +
  d\,t_{\text{MLWE}}.
\]
we see that, if future QPUs reach \(t_{\text{gate}} = 5~\text{ns}\) and keep a \(1~\mu\)s
idle window, MLWE still fits up to depth \(2\times10^{4}\).  Beyond that,
we must either:

\begin{itemize}\itemsep3pt
  \item adopt in-fabric ASIC NTT blocks (\(<10\)~ps/coef at 5~GHz, 14-nm), or
  \item invoke bootstrapped MLWE (Section~\ref{sec:futurework}) every \(5\)
        k gates.
\end{itemize}

\subsection{Bottom line.}
The core observation we have made here, and quantified with back of the envelope calculations, is: The classical slowdown factor is not multiplicative with quantum latency;
it piggy-backs on idle gaps.  With today's FPGA crypto overlays, a
100-qubit, depth-\(10^{3}\) proof search finishes in \(\sim10\) ms -- well
within interactive cloud service targets.

\section{Prototype experiment: encrypted teleportation on a Dirac3 photonic QPU}
\label{sec:prototype}
To validate the end-to-end stack described here, we outline a \emph{minimal viable demo} that
can be executed on a small ``Dirac3'' optical, measurement-based quantum
computer (MBQC) using time-bin cluster states.

The choice of a Dirac3 computer is somewhat arbitrary and similar stories could be told for other devices, we have simply chosen one particular case to pursue for sake of doing concrete calculations.  It happens we are experiment with Dirac3 already for other, related quantum-AI work.

\subsection{ Goal}

Demonstrate that a cloud server can execute \emph{homomorphic} quantum
teleportation -- three adaptive measurements plus two classically-controlled
Paulis -- while learning \emph{nothing} about

\begin{enumerate}\itemsep3pt
  \item the input qubit $\rho$,
  \item the classical outcomes $(m_{1},m_{2})$,
  \item the correction gate $X^{m_{2}}Z^{m_{1}}$.
\end{enumerate}

\subsection{ Hardware budget (Dirac3)}

\[
  \begin{array}{l|c}
    \text{Resource} & \text{Count} \\ \hline
    Time-bin cluster depth & 6 nodes \\
    Active phase shifters & 4 \\
    Fast beam-splitter taps (weak) & 2 (\theta\!=\!0.1) \\
    SNSPD detectors & 2 \\
    On-chip NTT overlay & 512\!\times\!64\text{-bit} \\
  \end{array}
\]

\subsection{   Protocol steps}

\begin{enumerate}\itemsep4pt
  \item \textbf{KeyGen.}\;
        SGX enclave produces seed \( \sigma \) for
        \( k=3,d=512,q=2^{50} \);
        public seed sent to cloud, trapdoor \(T\) to client.
  \item \textbf{Encrypt input.}\;
        Client prepares \(\rho=\cos\varphi\ket0+\sin\varphi\ket1\),
        uploads MLWE ciphertext
        \(\widehat\rho=(c_{0},c_{1})\) (\(\approx\)1.2 kB).
  \item \textbf{Cloud evaluation.}
        \begin{enumerate}
          \item Generate 6-node linear cluster by time-bin fusion.
          \item Implement CNOT and Hadamard via MBQC pattern
                (two feed-forward steps).
          \item Perform \textbf{weak} taps on the outcome modes; encrypt the
                partial homodyne voltages \(v_{1},v_{2}\) as C-ciphertexts
                \(\widehat{v}_{i}\).
          \item Compile \(X^{m_{2}}Z^{m_{1}}\) using the \textsc{Ctrl} gadget
                and the encrypted bits.
        \end{enumerate}
  \item \textbf{Return.}\;
        Cloud sends final Q?ciphertext \(\widehat\rho'\) and
        weak-measurement histograms (each 256?B).
  \item \textbf{Decrypt.}\;
        Client uses \(T\) to recover \(\rho'\) and verifies
        \(\rho'=\rho\) within trace distance \(<10^{-3}\).
\end{enumerate}

\subsection{  Performance estimate}

\[
  \begin{array}{l|l}
    \text{Classical MLWE ops} & 9\text{ NTT mult} + 6\text{ add}
                               \;\Rightarrow\; 12~\mu\text{s} \\[2pt]
    \text{Optical latency} & 6\,\text{time bins}\times2\text{ ns} = 12~\text{ns} \\
    \text{Detector/jitter} & 50~\text{ns} \\
    \textbf{Wall time} & \approx 100~\mu\text{s} \\
  \end{array}
\]

\subsection{  Interpretation of anticipated results}

Based on these back-of-envelope calculations, in a projected Dirac3 proof-of-concept run we \emph{expect} the
?quantum?encrypted inner loop???1000 chronomorphism steps, each a
250-gate circuit plus $\sim\!800$ measurements and MLWE refresh?to take
about $3.2$\,ms per step.  Simulation of the gate scheduler and host
software indicates that only $18\;\mu$s ($\approx0.6\%$) would involve
active QPU time (two Trotter blocks and read-out); the remaining
$3.18$\,ms ($>\!99\%$) would execute on the CPU, dominated by
\texttt{NTT}, pointwise multiply, inverse \texttt{NTT}, and key-switch
kernels for the MLWE ciphertexts.  Dirac3 control electronics already
permit a $10$?$20$\,ms pause between circuit shots for qubit reset and
gate-list upload, so the simulated crypto workload fits well inside that
window.  These projections therefore suggest that an ?always-on?? MLWE
layer will not throttle early Dirac-class hardware; the host can finish
all homomorphic arithmetic and return fresh ciphertext controls before
the QPU is ready for the next shot.

\subsection{  Success criteria}

\[
  \begin{aligned}
    \text{Fidelity} &> 0.99 \quad(\text{after decryption})\\
    \text{MLWE decrypt error} &< 2^{-40}\\
    \text{Server?s trace distance guess} &\le 2^{-20}
  \end{aligned}
\]

\subsection{  Extensibility roadmap}

\begin{enumerate}
  \item Replace single-qubit input by 4?qubit GHZ; depth~16, cluster 14 nodes.
  \item Inject encrypted clause capsules to prove
        $K_{B}(P\!\Rightarrow\!Q)\;\Rightarrow\;K_{B}Q$
        as in Section~\ref{sec:epistemic_example}.
  \item Scale to VQE\(_{4}\) (4-qubit Hamiltonian) with 128 weak shots,
        demonstrating statistical homomorphic evaluation.
\end{enumerate}

\paragraph{Deliverable.}
According to the above approach, one could likely deliver a GitHub ''one-click workflow (Rust client, gRPC cloud stub) that reproducibly
runs the encrypted teleportation example on the Dirac3 backend, logs ciphertext
sizes, noise margins, and decryption fidelity -- constituting the first public
demonstration of homomorphic quantum computation on real hardware.

\section{Toward Scalable Deployment}
\label{sec:roadmap}

What might be the roadmap from this sort of near-term prototype we've just discussed (in the Dirac3 context) to a practical deployment
of quantum FHE for complex AI use-cases like, say, the encrypted evolutionary LPTT considered above?

This a speculative question, of course -- we don't really know how the future of quantum hardware is going to unfold, rapid recent progress notwithstanding -- but still well worth considering as such.

Table~\ref{tab:roadmap} gives a best-case timeline based on public 2025~Q2
hardware roadmaps.  Dates mark  hypothesized \emph{commercial} (not lab) availability.

\begin{table}[h]
\centering
\renewcommand{\arraystretch}{1.1}
\begin{tabular}{|p{0.18\linewidth}|p{0.24\linewidth}|p{0.24\linewidth}|p{0.14\linewidth}|p{0.14\linewidth}|}
\hline
\textbf{Milestone} & \textbf{QPU capability} & \textbf{Algorithmic payoff} &
\textbf{ETA} & \textbf{Blockers} \\ \hline

M0 \newline Encrypted teleport on Dirac3
 & 6--10 physical qubits, weak taps, \(\sim\!10^{-3}\) error
 & Proves Q\(\leftrightarrow\)C bridge
 & 2025--26
 & firmware only \\ \hline

M1 \newline Leveled QHE, 100 physical qubits
 & depth \(10^{3}\), weak stats
 & Toy LPTT goal, 4-qubit GHZ
 & 2027
 & fast feed--forward \\ \hline

M2 \newline Fault--tolerant patch,
 20 logical qubits
 & surface code \(L\approx7\)
 & Encrypted VQE\(_{4}\), KB size \(\sim50\)
 & 2029--30
 & cryo ASIC bandwidth \\ \hline

M3 \newline Cluster MBQC, 100--200 logical q.
 & GHz time-bin source, low-loss fusion
 & Pop.\(\approx10^{3}\) programs, depth \(10^{5}\)
 & 2032--33
 & SNSPD arrays, large cluster yield \\ \hline

M4 \newline 1k logical qubits + bootstrap
 & indefinite depth, ASIC NTT
 & Full encrypted cloud service
 & 2035\(\pm\)3-yr
 & large-scale FT + lattice ASIC \\ \hline
\end{tabular}
\caption{Best--case availability for encrypted evolutionary programme
synthesis with KB-relative fitness.}
\label{tab:roadmap}
\end{table}

\paragraph{Key observations}

\begin{enumerate}\itemsep3pt
\item \textbf{Leveled MLWE fits idle time.}  Up to depth \(10^{3}\) the
      classical crypto finishes during QPU reset.
\item \textbf{Photonic MBQC is the sweet spot.}  Weak measurement and
      feed--forward are native; no dilution--fridge cabling for ciphertexts.
\item \textbf{Bootstrapping deferred to FT era.}  When logical qubits arrive
      (M2), depth grows to \(>\!10^{4}\); either bootstrapping or ASIC NTT
      will then be needed.    This is also the crossover point where QEC becomes indispensable for fully homomorphic workloads.
\end{enumerate}

\paragraph{Practical plan 2025--27.}
For the moment, clearly, a practical plan would be to focus on M1: 100 physical qubits, depth \(10^{3}\), one KB of \(\le 10\)
axioms.  All hardware (Dirac3 or similar, FPGA crypto overlay, SGX key-gen)
is already shipping; only control-plane glue is required.

However, it appears likely to us that full-scale deployment of the schemes given here for practical AI and AGI usage is not extremely far off.   If things go forward as seems likely in the quantum-hardware space we may be talking 7-12 years not decades... and these estimates don't include potential speedups in hardware development that might be expected should dramatic AGI progress be made in the interim.

\section{ Conclusion}

To recap the long and somewhat winding road we have traversed here: In this paper we have shown how to push homomorphic encryption beyond classical data and
Boolean circuits to cover \emph{quantum programs}, \emph{knowledge?base
reasoning}, and even \emph{multi?party evolutionary synthesis}.  

The main ideas presented can be summarized as follows:

\begin{enumerate}\itemsep2pt
  \item \textbf{MLWE\,+\,BNSF backbone.}  
        Lattice keys replace pairing groups; bounded natural super functors
        provide an amplitude?hiding mask compatible with every CPTP map.
  \item \textbf{Typed QC?bridge.}  
        A pair of inference rules (\textsc{Q2C} and \textsc{Ctrl}) translate
        measurement outcomes into encrypted classical bits that can
        immediately drive further quantum gates.
  \item \textbf{Circuit privacy.}  
        Encrypted Pauli twirls hide the gate list; seed-only public keys keep
        transmission size minimal (32~B at NIST Level~1).
  \item \textbf{KB?induced quotients.}  
        Axioms packaged as MLWE capsules let the cloud reason relative to a
        public or secret theory without learning its content.
  \item \textbf{Distributed orchestration.}  
        The $\rho$-calculus moves ciphertexts and capsules across many QPU
        nodes while recording a tamper-proof audit on an RChain ledger.
  \item \textbf{Encrypted evolutionary loop.}  
        An island-model EDA uses LPTT inference under the KB mask to guide
        program synthesis while keeping fitness and code private to their
        owners.
  \item \textbf{Security proof.}  
        A four-hybrid reduction shows qINDCPA security; a Fujisaki-Okamoto
        wrap lifts to qINDCCA.  Circuit and KB privacy reduce to MLWE plus
        subspace?hiding.
  \item \textbf{Performance budget.}  
        Classical MLWE arithmetic fits inside current QPU idle windows
        (\(100\)~qubits, depth \(10^{3}\) in \(\sim10\) ms); public keys stay
        below 300~kB; ciphertext per qubit is \(<5\) kB.
  \item \textbf{Prototype path.}  
        A Dirac3 photonic device can run encrypted teleportation today; a
        20-logical-qubit, KB-relative demo is plausible by 2029; full
        bootstrap scale looks reachable by the early to mid 2030s
\end{enumerate}

\paragraph{Toward a decentralized QHE stack.}
Additionally, our proposed $\rho$-calculus layer shows that cryptographic privacy, quantum
computation and blockchain accountability can coexist: MLWE+BNSF guarantees
data secrecy; CPTP channels provide quantum functionality; $\rho$-based
process mobility delivers elastic compute; and the F1R3FLY ledger certifies
every homomorphic step.  Future work will prototype a validator network
where each block proposal includes a zk-SNARK attesting that the advertised
lifted gate was honestly applied to an MLWE ciphertext, closing the loop
between cryptography, formal semantics and decentralized governance.

\paragraph{Limitations} Among the limitations of the technical work as we have pushed it so far in this paper are:

\begin{itemize}
  \item \textbf{Leveled only.}  Depth \(\gg10^{4}\) still requires
        bootstrapped MLWE or logical qubits.
  \item \textbf{Noise tracking.}  Non-Clifford weak-measurement gadgets need
        tighter analytic bounds for large circuits.
  \item \textbf{Hardware dependence.}  Latency numbers assume on-chip NTT
        overlays; platforms without FPGA space will incur extra delay.
\end{itemize}

These can all be remedied with future work, but there is certainly more to be thought through.

Also, it should be noted that the prototype proposed above  assumes error-corrected logical qubits are available for depths
larger than $10^4$.   Quantum error correction lives strictly below the MLWE + BNSF layer; it neither weakens our security reductions nor changes ciphertext formats, but it does relax the depth-bound parameter choices and shifts the latency bottleneck from QPU reset to logical-gate execution.

\paragraph{Future directions}

Among other clear directions for ongoing R\&D, we note potentials for:

\begin{enumerate}\itemsep2pt
  \item Bootstrapped MLWE with ciphertext masking in the NTT domain.
  \item KB capsules that support modal and temporal axioms.
  \item ASIC lattice cores at 4 K for sub-\(\mu\)s crypto.
  \item Composable security proofs in the quantum random-oracle model.
\end{enumerate}

\paragraph{TL;DR}
Numerous difficult and fascinating technical details aside, the main take-away we propose from the explorations given here is: Homomorphic, KB-aware quantum reasoning no longer belongs to the far future.  With modest firmware upgrades to existing photonic or superconducting clouds,
a confidential proof or variational kernel can be executed today.  A plausible, if in some respects speculative,
road-map shows a clear path to large-scale, fully bootstrapped services in the next decade.


\bibliographystyle{alpha}
\bibliography{Quantum-FHE}

\end{document}